\documentclass[letterpaper, 11pt]{article}

\usepackage[top=0.8in, bottom=0.9in, left=1in, right=1in]{geometry}
\usepackage[justification=centering]{caption}
\usepackage{setspace}

\spacing{1.1}
\usepackage{url}
\usepackage{times}
\usepackage{color}

\usepackage{amsfonts}
\usepackage{amsmath,amssymb,amsthm}
\usepackage{hyperref}
\usepackage{xspace}
\usepackage{dsfont}

\newtheorem{theorem}{Theorem}[section]
\newtheorem{lemma}[theorem]{Lemma}

\newtheorem{corollary}[theorem]{Corollary}

\newtheorem{claim}[theorem]{Claim}

\usepackage[ruled]{algorithm2e} % For algorithms

\SetAlFnt{\small}
\SetAlCapFnt{\small}
\SetAlCapNameFnt{\small}
\SetAlCapHSkip{0pt}
\IncMargin{-\parindent}

\def\calF{\mathcal{F}}
\def\rev{\textsc{Rev}} \def\drev{\textsc{DRev}}

\begin{document}
\title{On the Complexity of Simple and Optimal Deterministic Mechanisms for an Additive Buyer}

\author{Xi Chen\thanks{Columbia University. Email: \texttt{xichen@cs.columbia.edu}.
Research supported by NSF CCF-1149257 and CCF-1423100.}
\and George Matikas\thanks{Columbia University. Email: \texttt{matikas@cs.columbia.edu}.
Research supported by NSF CCF-1423100.}
\and Dimitris Paparas\thanks{University of Wisconsin-Madison. Email: \texttt{paparas@cs.wisc.edu}.
Research supported by NSF CCF-1320654 and CCF-1617505.}
\and Mihalis Yannakakis\thanks{Columbia University. Email: \texttt{mihalis@cs.columbia.edu}.
Research supported by NSF CCF-1320654 and CCF-1423100.}}
\date{}

% note that the abstract must come before \maketitle

\maketitle

\begin{abstract}
We show that the %problem of \edit{implementing a revenue-optimal deterministic mechanism} 
  {Revenue-Optimal Deterministic Mechanism Design problem}
for a single additive buyer is \#P-hard, even when the distributions have support size $2$ for each item and, more importantly, even when the optimal solution is guaranteed to be of a very simple kind: the seller picks a price for each individual item and a price for the grand bundle of all the items; the buyer can purchase either the grand bundle at its given price or any subset of items at their total individual prices. The following problems are also \#P-hard, {as immediate corollaries of the proof}:

\begin{enumerate}
\item determining if individual item pricing is optimal for a given instance,\vspace{-0.03cm}
\item determining if grand bundle pricing is optimal, and\vspace{-0.03cm}
\item computing the optimal (deterministic) revenue.
\end{enumerate}

On the positive side, we show that when the distributions are i.i.d. with support size $2$, the optimal revenue obtainable by any mechanism, even a randomized one, can be achieved by a simple solution of the above kind (individual item pricing with a discounted price for the grand bundle) and furthermore, it can be computed in polynomial time. 
{The problem can be solved in polynomial time too when the number of items is constant.}
\end{abstract}
\newpage

% !TEX root = main.tex

\def\vv{\mathbf{v}} \def\xx{\mathbf{x}} \def\srev{\textsc{SRev}\xspace} \def\brev{\textsc{BRev}\xspace}

\def\sprev{\textsc{SPRev}\xspace}

\def\sbrev{\textsc{SBRev}\xspace}

\section{Introduction}

{Consider the following natural scenario: A customer walks in a grocery store with the intention of buying some items. The store owner has statistical information from past customers that reveals how much 
%one might 
a typical customer values each 
%of the 
item. Her goal is to assign prices for the items and offer discounts for bundles of them to encourage %customers 
the customer to spend more money in a way that maximizes her expected revenue.

In this paper we formally study practices like the above, 
%a problem known in the literature as \emph{(Revenue) Optimal 
%Truthful 
%Deterministic Mechanism Design}, 
which we refer to as the \emph{optimal bundle-pricing} problem,
under the setting where a single \emph{additive} buyer is interested in $n$ \emph{heterogeneous} items offered by a seller.}
While the buyer's values for the items are unknown, the seller 
  {is given as input} a product distribution $\calF=\calF_1\times \cdots\times \calF_n$ from which the valuations $\vv=(v_1,\ldots,v_n)$ of the buyer for the $n$ items are drawn, where each $\calF_i$ is a discrete distribution given explicitly (by listing its support and probabilities).
The seller offers a finite \emph{menu} $M$ of \emph{bundles} to the buyer 
  {(or a \emph{bundle-pricing})}, %bundle-pricing})}
  %to the buyer,
  with each entry of the menu consisting of a subset (bundle) 
  $T\subseteq [n]$ of items and the price $\pi_T$ at which it is sold.
Given {a menu} $M$, the buyer {draws her valuations $\vv$ from $\calF$}
  and then either buys a bundle $T$ from $M$ that maximizes her utility
  $\sum_{i\in T} v_i-\pi_T$ or nothing if the utility of every bundle in $M$ is negative\footnote{Ties in utility are broken in favor of a bundle with higher price (and value).}; 
  the price $\pi_T$ of the bundle bought by the buyer is the \emph{revenue} of the seller.
The goal of the seller is to find a menu that maximizes her
  expected revenue (i.e., the expected price $\pi_T$ that the buyer pays), which is known to be equivalent to the problem of 
  \emph{Revenue Optimal Deterministic Mechanism Design}. 
{{If we extend bundle-pricings %by offering \emph{lotteries}  
 % (where
%footnote{In the 
%lottery-pricing problem
%{\color{blue}Optimal Randomized Mechanism Design problem}, 
  to allow the seller to offer 
  a finite menu of lotteries (or a \emph{lottery-pricing}), where %\emph{lottery-pricing}): 
  a lottery is a pair $((x_1,\ldots,x_n),\pi)$ with $\pi$ being its price
  and  $x_i\in [0,1]$ being the probability of the buyer getting  
  item $i\in [n]$ if this lottery is purchased,
 we obtain the \emph{optimal lottery-pricing problem}, also known 
  in the literature as the problem of \emph{Revenue Optimal
{Randomized Mechanism Design}}.}}  
  We define the two problems formally in Section \ref{sec:preliminaries}.\footnote{{We choose to follow the language of bundle-pricings and lottery-pricings in this paper, instead of
  deterministic and randomized mechanisms, mainly because they are conceptually
  closer to common practices seen in the real world and are easier to 
  understand for readers that are not familiar with mechanism design (e.g., there is no need to
  introduce the notion of truthfulness).}}
%The goal of the seller is to maximize her expected revenue, i.e., the expected

It is worth pointing out that bundle-pricing schemes commonly used in
  practice do not necessarily~list explicitly the bundles offered in the menu,
%  Note that the seller does not necessarily need to list explicitly all bundles offered in the menu, \edit{which in principle can be exponentially many}, 
   but may specify them implicitly in a succinct manner.  
   For example, in the case of offering a simple \emph{item-pricing} without any discounts for bundles,
   the seller needs only to specify a price for each item ($n$ numbers in total); the induced menu
   consists of all $2^n$ subsets of items, each priced at the sum of 
   prices of items in the subset, and has the desired property that the buyer's problem,
   i.e., finding an optimal bundle in the menu given $\vv$, is easy to solve.
{Because of this, it would not be appropriate to require the output of
  the bundle-pricing problem
  to be an explicit list of bundles in an optimal menu, but rather
  it can be represented in a reasonably succinct way.
The exact representation of the output, however, will not affect our main results, as we explain when we describe them later
in this Section.}

Both optimal %Bundle-pricing and Lottery-pricing problems 
{deterministic and randomized mechanism design problems}
have been studied intensively
  during the past decade \cite{Than,GHK+05,Briest08,ChK07,ChHMS10,CD11,ChMS10,BCKW10,pavlov,WangT14,Bab14,DDT12a,DDT12b,DDT13,DDT15,HN12,ManelliV06,Rub16,Li13,Yao15}. 
{For some instances, randomized mechanisms can achieve strictly higher revenue than deterministic mechanisms. However, deterministic mechanisms (bundle-pricings) are much more widely used in practice (especially ``simple'' pricing schemes); we will focus on deterministic mechanisms in this paper.}
Recently, much effort has been devoted to understanding the power and limitations
  of simple pricing schemes, that is, menus~that can be described succinctly in a natural way
  {and at the same time induce an easy-to-solve buyer's problem}.
Some of the examples include (i) selling all items separately (item-pricing), (ii) selling only the grand bundle that consists of all items (grand-bundle pricing),
   and (iii) partition mechanisms, where one partitions the items into disjoint groups,
   each with its own price, and sells the groups separately. 
While it is known that none of these solutions is optimal in general among
  bundle-pricings,
  %{\color{blue}bundle (or lottery) pricings},
there has been substantial work studying basic questions for
each of these simple solutions, including the following: 
How does the revenue achievable by these solutions compare with optimal 
   revenues achievable by %a bundle-pricing or a lottery-pricing?}
   bundle or lottery pricings?
What are conditions under which these solutions are optimal? 
Can we compute an optimal solution of each type?

In case (i) of selling the items separately, we know how
to compute efficiently an optimal item-pricing: each item is assigned
separately its optimal price following Myerson's theory \cite{M81}.
In both cases (ii) and (iii) of the grand bundle and partition mechanisms,
the problem of finding an optimal solution is intractable 
(\#P-hard \cite{DDT12b} and NP-hard \cite{Rub16}, respectively).
However, the fact that it is hard to find an optimal solution
of a certain type (grand bundle or partition mechanisms) 
%in an instance 
  does not mean that one cannot easily find
  a solution that is not of this type {and has higher revenue
(for example, by selling also individual items), or possibly even}  find a solution that is optimal among all bundle-pricings.
 %find better solutions that are not of this type. 
%Indeed, in the constructions for both of these hardness results, one can increase the revenue by simply offering also separately the items at appropriate prices (and this is a frequent phenomenon, i.e. not just for these instances).
Thus, two central questions remain concerning the bundle-pricing (or optimal deterministic
  mechanism design) problem:
{
\begin{enumerate}
\item Is there an efficient algorithm that finds an optimal bundle-pricing?\vspace{-0.08cm}
\item %At least, 
{If the problem above is hard in general},
is there such an algorithm when the instance is promised\\ to have a ``simple'' optimal 
  bundle-pricing?
\end{enumerate}
}

{Our results resolve {both} questions in the negative by showing that the problem is \#P-hard, even when the distributions have support size $2$ for each item and, more importantly, even when 
  the instance is promised to have a unique optimal bundle-pricing that is 
  of a very simple kind,
%the optimal solution is guaranteed to be of a very simple kind, 
which we call a \emph{discounted item-pricing}: the seller picks a price for each individual item and a price for the grand bundle of all the items; the buyer can purchase either the grand bundle at its given price or any subset of items at their total individual prices.} 
Such a solution can be described using $n+1$ numbers and the buyer's problem is also easy to solve.
{This is the reason why the exact output format of the problem does not 
  affect our hardness result.}
%This addresses in the negative the second question for this family of simple mechanisms as well. 

{This result tells us that the bundle-pricing (deterministic mechanism design) problem is inherently computationally hard, and furthermore the difficulty is not (only) due to the fact that the optimal solution can be very complex, of the kind that one would not use in practice anyway; the problem is hard even when the optimal solution is extremely simple: standard item pricing with a discount for the grand bundle.}

As a by-product of the proof, we also resolve in the negative the question
of whether there is a `nice' characterization (i.e., an easy-to-check
necessary and sufficient condition) of when item-pricings 
are optimal,~i.e.,~whether~an item-pricing can achieve the optimal revenue achievable by
  bundle pricings or whether bundling helps.
The same applies to grand-bundle pricings (and partition mechanisms),
i.e.,~there~is no easy-to-check  
characterization for the optimality of grand-bundle pricings under standard complexity-theoretic~assumptions.

{On the positive side, we show that when $\calF_1,\ldots,\calF_n$ are i.i.d. with support size 2, the optimal revenue achievable by any pricing scheme, even a lottery one, can be achieved by a discounted item-pricing~which, furthermore, can be computed in polynomial time. We discuss our results in detail below in Section \ref{ourresults}.}

\subsection{Our Results}\label{ourresults}

Given an input distribution $\calF$, we use
  $\brev(\calF)$, $\srev(\calF),\drev(\calF)$ and $\rev(\calF)$ to denote the optimal expected revenues %(nondecreasing from left to right) 
   achievable by a grand-bundle pricing (i.e., selling the grand bundle only),
   an item-pricing (i.e., selling all items separately),   
 % {\color{blue}deterministic mechanisms, and randomized mechanisms}. 
   a bundle-pricing, and a lottery-pricing, respectively. %

%We obtain both algorithmic and hardness results related to discounted item-pricing mechanisms.
First %, \edit{regarding the second open question stated in the introduction} 
%we show that when the  distributions $\calF_1,\ldots,\calF_n$ are i.i.d. 
%  and have support size $2$,
%  the optimal revenue $\rev(\calF)$ from %lottery-pricings (
%  randomized mechanisms can always be achieved
%  by a discounted item-pricing mechanism, which can also be computed efficiently in polynomial time.
  we state our positive result for i.i.d. distributions with support size $2$:

\begin{theorem}\label{thm:iid}
%The problem of computing $\emph{\rev}(\calF)$ can be solved in polynomial time
\hspace{-0.05cm}When $\calF_1,\ldots,\calF_n$ are i.i.d. with support size $2$ \emph{(}$\{a,b\}$ and $a<b$\emph{)},
  $\emph{\rev}(\calF)$ can~always be achieved by a discounted item-pricing where
  the grand bundle is priced at $kb+(n-k)a$ for some $k\in [0:n]$ and
  each item is priced at $b$.
Moreover, the parameter $k$ can be found in polynomial time.
%  Moreover, there exists an integer $k\in [0:n]$ such that $\emph{\rev}(\calF)$
%  is achieved by offering the following bundles (letting $\{a,b\}$ be the support with $a<b$):
%\begin{enumerate}
%\item The grand bundle $[n]$ at price $kb+(n-k)a$; and
%\item For each $T\subseteq [n]$ of size at most $k-1$, the bundle $T$ at price $|T|b$ (or equivalently,
%  each item\\ is priced at the high value $b$ and the buyer can buy any bundle of size up to $(k-1)$).
%\end{enumerate}
\end{theorem}

%$\calF=\calF_1\times \cdots\calF_n$. 
%We use $\rev(\calF)$ to denote the optimal revenue achievable
%  by any menu of lotteries (or any randomized auction) when selling $n$ items to 
%  a single additive buyer whose values for items are 
%  distributed according to the product distribution $\calF$.
%We also use  
%  $\drev(\calF)$ to denote the optimal revenue achievable
%  by any bundle pricing (or any deterministic mechanism).
%See the formal definitions in Section \ref{sec:preliminaries}. 

{Our main result addresses the {two questions} from the introduction in the negative:  we show that it is \#P-hard to find an optimal bundle-pricing % a revenue optimal deterministic mechanism 
  for a single additive buyer, even when $\calF$ is 
  a product distribution and each $\calF_i$ has support size 2. 
Although in general an optimal solution can be highly complex and consist of exponentially many bundles without a succinct description,
our hardness result is established on instances that are guaranteed to have a unique and very simple optimal solution, 
%and has further implications regarding a natural class of mechanisms that can be described succinctly, 
namely, the discounted item-pricing that we defined earlier. Such a pricing scheme  corresponds to the ubiquitous practice of offering an individual price $\pi_i$ for each item $i$ 
%in a group of complementary goods $T$ 
and also the grand bundle of all items at a discounted price $\pi$, for example a combo of a toothpaste, a toothbrush, dental floss, and mouth wash offered at a 15\% discount as compared to the cost of buying them separately. The buyer can choose to buy the grand bundle at $\pi$ or any subset $T$ of items  at $\sum_{i\in T} \pi_i$, whichever brings the highest (nonnegative) utility
{(note that in the latter case the buyer will obviously buy 
the set $T$ of all items whose price is less than or equal to the buyer's value)}. While a discounted item-pricing offers exponentially many bundles, it has a succinct representation by $n+1$ numbers and is easy to implement in practice.} 
{We state our main hardness result in Theorem \ref{thm:main}:}
%Notice that this includes the possibility~of offering the best of either only the grand bundle or only the items separately. The latter was studied by Babaioff et al. \cite{Bab14} where they proved that, for any product distribution $\calF$, such a mechanism (and hence discounted item-pricings) extracts at least 1/6th of the revenue that the best randomized mechanism can achieve.}

\begin{theorem}\label{thm:main}
%\editx{It is \#P-hard to implement the revenue-optimal truthful deterministic mechanism} 
{The optimal bundle-pricing problem is \#P-hard even when (1) all
  distributions have support size $2$ and (2)
  the instance is promised to have a unique optimal solution 
  that is a discounted item-pricing}. 
  %with one of the
  %following two forms:
  %the {unique} optimal solution is guaranteed to have one of the following two forms: 
  %either (i) 
  %separate item
%pricing 
  %the {\color{blue}grand-bundle pricing that offers the grand bundle at the 
  %sum of low values of all items}\footnote{Note that a grand-bundle pricing 
  %is a special case of discounted item-pricings.}, or (ii) a discounted item-pricing   mechanism with all items priced at their high values
%and the grand bundle priced at a specific value.
\end{theorem}  \vspace{0.08cm}
 
{\noindent Indeed, the hard instances constructed in the proof of Theorem \ref{thm:main} have the 
  property that either 
\begin{enumerate}
\item[(i)]  
the grand-bundle pricing\footnote{Note that a grand-bundle pricing is a special case
  of discounted item-pricings.} that offers the grand bundle at the sum of low values of all items;
or \vspace{-0.1cm}
\item[(ii)] the discounted 
  item-pricing that offers each individual item at its high value and 
  the grand bundle\\ at a specific value that can be computed from the instance in polynomial time,
\end{enumerate}
%  with the  grand bundle priced at a value that can be computed from
%  the input instance in polynomial time 
is guaranteed to be optimal among all bundle-pricings, but it is \#P-hard
to determine which one is better.}
Note that %in general, item-pricings are also (degenerate) special
%cases of discounted item-pricing (by setting the price of the grand bundle
%equal to the sum of all item prices).
%In the above case, solution 
  {(i) can be equivalently described as
  an item-pricing, with each item priced at its low value, and the revenue can be computed 
  in polynomial time.
  %the item-pricing (i) can be equivalently replaced
%by the grand bundle solution that simply offers 
%the grand bundle at the sum of low values of all items.
These observations together lead to a number of 
  corollaries.}

\begin{corollary}\label{coro1}
The following problems are \#P-hard:
\begin{flushleft}\begin{enumerate}
\item Given a product distribution $\calF$, decide  whether $\emph{\drev}(\calF)=\emph{\srev}(\calF)$, i.e., %whether {\color{blue}an optimal deterministic mechanism is strictly better
 % than all item-pricing mechanisms;}
{whether an item-pricing is optimal among all bundle-pricings.}\vspace{-0.1cm}
 % the optimal bundle-pricing is strictly better than selling the items separately.\\
\item Given a product distribution $\calF$, decide whether $\emph{\drev}(\calF)=\emph{\brev}(\calF)$, i.e., 
{whether a grand-bundle pricing is optimal among all bundle-pricings.}%whether {\color{blue}an optimal deterministic mechanism is strictly better 
%  than all grand-bundle mechanisms.}
   %the optimal bundle pricing is strictly better than selling the grand bundle only.
\end{enumerate}\end{flushleft}
\end{corollary}

\begin{corollary}\label{coro2}
The following problems are \#P-hard: 
\begin{enumerate}
\item Given a product distribution $\calF$, compute $\emph{\drev}(\calF)$.\vspace{-0.12cm}
\item Given a product distribution $\calF$ and
a valuation $\vv$, compute the bundle bought at $\vv$ in any\\ optimal bundle-pricing.
\end{enumerate}
\end{corollary}

%Another immediate consequence of Theorem \ref{thm:main}
%is that the problem of computing the expected revenue
%of a given discounted item-pricing is \#P-hard.
%By another reduction however, we can show the stronger
%fact that this holds even for grand bundle pricing.

%\begin{theorem}\label{thm:revenue}
%The following problem is \#P-hard:
%Given distributions $\calF_1,\ldots,\calF_n$ and
%a price $\pi$ for the grand bundle, which is in fact the optimal
%grand bundle price, compute its expected revenue.
%\end{theorem}

%As mentioned above, computing the optimal grand bundle
%solution is \#P-hard \cite{DDT12b}. Using the techniques of
%\cite{Rub16} we can show however that one can compute
%an approximately optimal solution.

%\begin{theorem}\label{thm:ptas}
%Given a (discrete) product distribution $\calF$ and $\epsilon>0$,
%we can compute in polynomial time a grand bundle pricing
%that has expected revenue at least $(1-\epsilon) \emph{\brev}(\calF)$.
%\end{theorem}

We remark finally that all the hardness results hold if
the number of items is unbounded.
For a constant number of items, we obtain a polynomial-time
algorithm {(though the dependency of its running time on the number of items is exponential)}.%the polynomial grows with the number of items.

\begin{theorem}\label{thm:constant}
When the number of items is constant, 
  an optimal bundle-pricing can be computed in polynomial time.
%an optimal menu of bundles can be computed in polynomial time.
\end{theorem}

\subsection{Related Work}

The seminal work of Myerson \cite{M81} completely settles the 
%single-parameter case, where there is one item for sale, 
  case of selling a single~item,
by giving a {computationally efficient} and {deterministic} mechanism {(i.e., a pricing of the item}) that maximizes the expected revenue among all possible, randomized or deterministic, mechanisms. The more general {multi-dimensional} setting, however, turns out to be inherently more difficult. Unlike Myerson's setting, randomization in~general improves the revenue when there are many items for sale, even if there is a single unit-demand buyer \cite{Than} {(i.e. the buyer wants to buy only one item)} or an additive buyer \cite{ManelliV06}.~It~is~also~known~that~the optimal menu of lotteries may have exponential size \cite{CD+15,HartNisan13}. Moreover, under standard complexity-theoretic assumptions, recent results rule out the existence of computationally efficient algorithms that find a revenue-optimal deterministic or randomized mechanism for a unit-demand buyer \cite{CD+14,CD+15,GHK+05,Briest08}, or a randomized mechanism for an additive buyer \cite{DDT12a}.
However, hardness \mbox{results} for the optimal deterministic mechanism design problem with an
  additive buyer are limited. 
Rubinstein \cite{Rub16} proved that finding an optimal partition mechanism is strongly NP-hard;  Daskalakis et al. \cite{DDT12b} proved that finding an optimal price for selling the grand bundle is \#P-hard. These results are for restrictions of the problem that impose a   specific menu structure, % and this is where the hardness stems from. 
%Therefore, they do not extend to the general case of deterministic additive pricing, which, in principle, could be %an easier problem and its complexity 
   and the original problem remained open  before this work.

{Among these hardness results, the one that is most relevant to ours is that of 
  Daskalakis, Deckelbaum, and Tzamos \cite{DDT12a}.
They construct instances $\calF$ with support size $2$ for each $\calF_i$ and 
  show that the problem of finding~an optimal lottery-pricing (or randomized mechanism) is \#P-hard.
This, however, {does not have any consequences} for the bundle-pricing problem for two reasons.
First, the deterministic mechanism design problem is not necessarily harder than 
  the randomized one.
In fact, in the setting of a unit-demand buyer, the deterministic problem is provably ``easier'': the randomized problem is \#P-hard \cite{CD+15} while the deterministic one is in~NP \cite{CD+14}.
Second, for the construction of \cite{DDT12a} to work for bundle-pricings,
  it would need to be the case that 
  optimal menus of lotteries of their instances are deterministic and
  consist of bundles only.
{However, this is not the case: the solution in \cite{DDT12a} makes essential use of the randomization feature and
%A close look  shows that this is not the case: in addition to offering the 
%  grand bundle~and individual prices for items (as in a discounted item-pricing),
 the optimal menu for $\calF$  contains 
  a large number of lotteries (with probabilities in $(0,1)$) 
  for valuations in a certain ``critical'' region.}
Compared to techniques used~in~\cite{DDT12a}, ours are different  in the following two aspects: (1) 
  Since $\drev(\calF)$ is captured by an integer program (instead of a linear program for
  $\rev(\calF)$; see Section \ref{sec:preliminaries}), we cannot use LP duality but have to rely on more discrete and combinatorial arguments to identify its optimal integer solutions; (2) 
  An important step in both proofs is to relax the integer (or linear) program
  that captures $\drev(\calF)$ (or $\rev(\calF)$). 
Our relaxation is significantly different from the LP relaxation of \cite{DDT12a}. 
We need to keep a large set of global constraints from the original IP 
  while local constraints suffice for the purpose of \cite{DDT12a}. 
%It seems to us that the latter does not suffice in our setting.
}

%Those settings however do not include the most common scenario of giving deterministic~discounts for bundles of items that we described earlier in the introduction; the corresponding setting for which~is that of a deterministic mechanism for an additive buyer. 
%Other than the hardness results,
Most of the work on the deterministic mechanism design problem for
  an additive buyer so far focuses on approximation. Hart and Nisan \cite{HN12} studied two simple deterministic mechanisms for product distributions: selling items separately or selling the grand bundle~only. They showed that selling items separately and grand bundling are respectively $\smash{\Omega(1/\log^2 n)}$ and $\smash{\Omega(1/\log n)}$ approximations of the optimal  revenue achievable
  by any (possibly randomized) mechanism (later improved by Li and Yao \cite{Li13} to $\Omega(1/ \log n)$ for both schemes, which is known~to be tight \cite{HN12}).~While neither of these two schemes can achieve~by itself a constant factor approximation, Babaioff et al. \cite{Bab14} showed that %for product distributions 
   the \emph{better} of the two gives a $(1/6)$-approximation. Recently, Daskalakis et al. \cite{DDT13,DDT15} studied conditions for grand-bundling mechanisms to be optimal (for continuous distributions), and showed that this happens if and only if two stochastic dominance conditions hold. Rubinstein \cite{Rub16} worked on partition mechanisms and obtained a polynomial-time approximation scheme (PTAS) for a revenue maximizing partition mechanism. 
A number of other results \cite{symmetry,CH13} obtained approximation schemes for i.i.d. distributions with the MHR property. Giannakopoulos and Koutsoupias \cite{GK14} obtained optimal mechanisms for i.i.d. uniform distri\-butions with up to six items.   Finally Yao \cite{Yao15} introduced a new approach for reducing the $k$-item $n$-bidder problem to the $k$-item $1$-bidder setting and gave a deterministic mechanism that yields at least a constant fraction of the optimal revenue for the 
     more general $k$-item $n$-bidder setting. %auctions with arbitrary independent valuation distributions.

We also note that there is extensive work studying unit-demand buyers (e.g. \cite{CD+14,CD+15,GHK+05,Briest08,ChK07,ChHMS10,CD11,ChMS10,BCKW10,pavlov,WangT14}). 
Besides the papers cited earlier that address the complexity of  an optimal mechanism  in that context, 
the rest of the work, which mostly concerns 
special cases or approximation, is not directly related to the topic of the present paper.

%Closing this section, we note for completeness that there is a different line of work studying Unit-demand buyers. On top of the hardness results we mentioned earlier, a lot of this work focuses on algorithms and approximation. For the deterministic setting Chawla et al.\cite{ChK07} show that techniques from Myerson's work can be used to obtain an analogous closed-form characterization for prices that extract a factor $3$ approximation of the optimal expected revenue (subsequently improved to $2$ \cite{ChHMS10}). If however the distributions are restricted to be monotone hazard-rate, Cai and Daskalakis~\cite{CD11} obtained a polynomial-time approximation scheme.
%For the randomized setting, Chawla et al. \cite{ChMS10} showed that for product distributions randomization can improve the optimal revenue by at most a factor of $4$. In contrast, Briest et al. ~\cite{BCKW10} showed that when $\mathcal{F}$ is correlated the improvement in revenue by randomization can be unbounded even for instances with four items. However, if $\mathcal{F}$ is given explicitly, by listing the probability of every valuation vector, then one can solve a linear program and find an optimal mechanism in polynomial time. Pavlov \cite{pavlov} completely characterized optimal mechanisms when there are two items and their values are drawn independently from distributions that meet certain conditions. Finally, Wang and Tang \cite{WangT14} studied conditions under which the optimal randomized mechanism has ``simple" menus.

% !TEX root = main.tex
\def\uu{\mathbf{u}} \def\ww{\mathbf{w}}

\section{Preliminaries}\label{sec:preliminaries}

Let $D_i$ be the support of $\calF_i$, and  
  $D=D_1\times \cdots\times D_n$ be the set of valuation vectors. 
For each $\vv\in D$, let
$$
\Pr[\vv]=\Pr_{\calF_1}[v_1]\times\cdots\times\Pr_{\calF_n}[v_n]
$$
denote the probability of $\vv$ drawn from $\calF$.

We first define $\drev(\calF)$, the optimal expected revenue obtainable by 
  a bundle-pricing, by formulating it using an integer program with $n+1$ variables
  associated with each valuation $\vv\in D$: $x_{\vv,1},\ldots,x_{\vv,n}$~and $\pi_\vv$,
  where $x_{\vv,i}\in \{0,1\}$ indicates whether item $i$ is included in the bundle
  the buyer chooses from the menu (with $x_{\vv,i}=1$ if item $i$ is included) 
  when her valuation is $\vv$ and $\pi_\vv$ denotes 
   the price of the bundle.
%  
%  denotes the probability of the buyer receiving item $i$ 
%  in the lottery she picks from the menu~when her valuation is $\vv$ and $\pi_\vv$ denotes 
%  the price of the lottery. 
We~also write $\xx_\vv=(x_{\vv,1},\ldots,x_{\vv,n})\in \{0,1\}^n$ to denote the allocation vector for 
  valuation $\vv$.
The integer program then maximizes the expected revenue:
$ 
\sum_{\vv\in D} \pi_\vv\cdot \Pr[\vv]
$ 
subject to the following constraints:
\begin{enumerate}
\item $x_{\vv,i}\in \{0,1\}$ for all $\vv\in D$;\vspace{-0.08cm}
\item For each $\vv\in D$, the utility is nonnegative: $\sum_{i\in [n]} v_i\cdot x_{\vv,i}-\pi_\vv\ge 0$;\vspace{-0.08cm}
\item For all $\ww,\vv\in D$, $\ww$ does not envy the bundle of $\vv$:
$$
\sum_{i\in [n]} w_i\cdot x_{\ww,i}-\pi_\ww\ge
  \sum_{i\in [n]} w_i\cdot x_{\vv,i}-\pi_\vv.
$$
\end{enumerate}
We refer to this integer program as the \emph{standard IP} for $\drev(\calF)$ 
{and the goal of the optimal bundle-pricing problem
  is to find an optimal solution to the standard IP.
As discussed earlier, the exact way of defining the output of the problem does not affect 
  our main results.
(For example, one can adopt the model used in \cite{DDT12a,CD+15}, where a polynomial-time algorithm $A$ 
  for the optimal bundle-pricing problem takes as input a distribution $\calF$ and a valuation $\vv\in D$
  and outputs a bundle $A(\calF,\vv)$ such that $\{A(\calF,\vv):\vv\in D\}$ is an optimal solution to the standard  IP for $\drev(\calF)$.
Under this formulation Theorem \ref{thm:main} implies that there cannot be any such polynomial-time
  algorithm unless \#P can be solved in polynomial time.)}

{The equivalence between the optimal bundle-pricing problem and 
  deterministic mechanism design follows from the observation
  that any feasible solution $\{\xx_\vv,\pi_\vv:\vv\in D\}$ 
  to the standard IP for $\drev(\calF)$ can be equivalently viewed as a deterministic mechanism that is both 
  individually rational and truthful, and vice versa: the mechanism, upon $\vv$ reported by
  the buyer, assigns items $\xx_\vv$ to the buyer and charges her $\pi_\vv$.}

Sometimes (e.g., in Section \ref{sec:hard}), it is more convenient to replace $\pi_\vv$ by a nonnegative
  utility variable $u_\vv$.
The standard IP maximizes the same expected revenue:
$$
\sum_{\vv\in D} \left(\sum_{i\in [n]}v_i\cdot x_{\vv,i}-u_\vv\right)\cdot \Pr[\vv]
$$
subject to the following (slightly simpler) constraints:
\begin{enumerate}
\item $x_{\vv,i}\in \{0,1\}$ and $u_\vv\ge 0$ for all $\vv\in D$;\vspace{-0.08cm}
\item For all $\ww,\vv\in D$, $\ww$ does not envy the bundle of $\vv$:
$$
u_\ww\ge
  \sum_{i\in [n]} w_i\cdot x_{\vv,i}-\left(\sum_{i\in [n]}v_i\cdot x_{\vv,i}-u_\vv\right)
  =u_\vv+\sum_{i\in [n]} (w_i-v_i)\cdot x_{\vv,i}.
$$
\end{enumerate}
We refer to this IP as the standard IP (utility version) for $\drev(\calF)$.

On the other hand, the optimal revenue $\rev(\calF)$ obtainable by 
  a lottery-pricing is captured by the same objective function and linear constraints,
  except that $x_{\vv,i}$ takes values in $[0,1]$ instead of $\{0,1\}$.
We refer to this linear program as the standard LP for $\rev(\calF)$.
% !TEX root = main.tex

\section{IID with Support Size 2}\label{sec:iid}

\newcommand{\Exp}{\operatornamewithlimits{\mathbb{E}}}
\newcommand{\cube}{\operatorname{\{0, 1\}}}
\newcommand{\pcube}{\operatorname{\{\pm1\}}}

\newcommand{\F}{\mathbb{F}}
\newcommand{\C}{\mathbb{C}}
\newcommand{\R}{\mathbb{R}}
\newcommand{\Z}{{\mathbb Z}}
\newcommand{\N}{{\mathbb N}}
\newcommand{\calf}{{\cal F}}
\newcommand{\calm}{{\cal M}}
\newcommand{\eqdef}{{\stackrel{\rm def}{=}}}
\newcommand{\Fp}{{\mathbb F}_{p}}
\newcommand{\polylog}[1]{\mathrm{polylog}(#1)}
\newcommand{\poly}[1]{\mathrm{poly}(#1)}
\newcommand{\diff}[1]{\mathrm{diff}(#1)}
\newcommand{\wt}[1]{\mathrm{wt}(#1)}
\newcommand{\ith}[1]{{#1}^{\text{th}}}

\newcommand{\pairs}{\mathrm{pairs}}
\newcommand{\social}{\mathrm{social}}
\newcommand{\ud}{\mathrm{d}}
\newcommand{\pr}{\mathrm{P}}
\newcommand{\mmax}{\mathrm{Max}}

\newcommand{\dominates}{\preccurlyeq}
\newcommand{\dominated}{\succcurlyeq}

\def\rr{\mathbf{r}} \def\calC{\mathcal{C}} \def\xx{\mathbf{x}} \def\11{\mathbf{1}}
\def\pr{\textsf{Pr}} \def\vv{\mathbf{v}} \def\uu{\mathbf{u}} \def\cc{\mathsf{CC}}
\def\xx{\mathbf{x}} \def\yy{\mathbf{y}} \def\lpr{\mathsf{L}} \def\rpr{\mathsf{R}}
\def\ext{\mathsf{Ext}} \def\aveL{\mathsf{aveL}} \def\aveR{\mathsf{aveR}}
\def\minL{\mathsf{minL}} \def\minR{\mathsf{minR}} \def\ss{\mathbf{s}}
\def\den{\mathsf{den}} \def\upper{\mathsf{upper}} \def\low{\mathsf{lower}}
\def\zz{\mathbf{z}} \def\denL{\mathsf{denL}} \def\denR{\mathsf{denR}}
\def\calP{\mathcal{P}} \def\pp{\mathbf{p}} \def\qq{\mathbf{q}} \def\ww{\mathbf{w}}
\def\calM{\mathcal{M}} \def\calE{\mathcal{E}} \def\calD{\mathcal{D}}
\def\supp{\textsc{supp}} \def\brev{\textsc{BRev}} \def\calF{\mathcal{F}}
\def\brevh{\textsc{BRev}_\text{H}} \def\calN{\mathcal{N}}
\def\brevl{\textsc{BRev}_\text{L}} \def\calT{\mathcal{T}} \def\Rev{\mathsf{Rev}}
%\spacing{1.2}

We establish Theorem \ref{thm:iid} in this section.
Let $\calF_1,\ldots,\calF_n$ be i.i.d. distributions with support size~$2$.
%\begin{document}
%\begin{flushleft}\setlength{\parindent}{15pt}
%We have $n$ items with the same probability distribution with support 2.
Without loss of generality we can assume that the support is $\{1,b\}$ with $b>1$.
(If the support is $\{0,b\}$ the problem is trivial: the optimal revenue can be achieved
  by offering every item at price $b$; if the support is $\{a,b\}$ with $0<a<b$,
then we can equivalently rescale it to $\{1,b/a\}$.)
Let $p$ $\in (0,1)$ be the probability that an item takes value $b$, and $1-p$ that it takes 1.
We let $\calP_i$ denote the probability of $\vv\sim \calF$ having $i$ items
at value $b$ and $n-i$ at $1$, for each $i\in [0:n]$.
That is,
$$\calP_i = \binom{n}{i} \cdot p^i \cdot (1-p)^{n-i}.$$
The following lemma for $\calP_i$'s is crucial.
We delay its proof {to the end of this section}.

\begin{lemma} \label{lem:mon}
There exists an  integer $k\in [0:n]$ such that 
\begin{equation}\label{eq:jaja}
(n-i)\calP_i-(b-1)(\calP_{i+1}+\cdots+\calP_n)
\end{equation}
is negative for all $i: 0\le i<k$ and is nonnegative for all $i:k\le i\le n$.
%$\frac{(n-i) \calP_i}{\sum_{j\geq i+1} \calP_j}$ is a monotone increasing function of $i$.
\end{lemma}

Let $k\in [0:n]$ be an integer that satisfies Lemma \ref{lem:mon},
  which is unique and can be computed~in polynomial time.
%smallest index such that
%$(n-k) \calP_k \geq (b-1)\sum_{j\geq k+1} \calP_j$.
We use $S^*$ to denote the following discounted item-pricing:
\begin{quote}
The grand bundle $[n]$ is offered at $kb+n-k$ and
each item is offered individually at $b$ (the latter means that
 the buyer can buy any bundle $T\subseteq [n]$
  at price $|T|b$).
%For each $T\subseteq [n]$ of size at most $k-1$, the bundle $T$ at price $|T|b$. \edit{Equivalently,
%  each item is priced at the high value $b$, i.e., the buyer can buy any bundle of size $i$ at price $ib$; note %that if $i \geq k$ then
%  the buyer will prefer instead the grand bundle.}
\end{quote}
%which offers the grand bundle at price
%$kb+n-k$ and each individual item at price $b$.
Given $S^*$, the behavior of the buyer is as follows.
If a valuation vector has $k$ or more items at $b$ then the buyer buys the
grand bundle at $kb+n-k$; otherwise it buys all the items that have value $b$.
The expected revenue $R^*$ of the discounted item-pricing $S^*$ is then
$$
R^* = \sum_{1\leq i <k} b i \cdot\calP_i + (kb+n-k) \sum_{k\leq i \leq n} \calP_i.$$
It is clear that given $k$, $R^*$ can be computed in polynomial time.

To finish the proof of Theorem \ref{thm:iid},
  we show that $S^*$ achieves the optimal revenue $\Rev(\calF)$.
%We will show that $S^*$ is an optimal solution, including among randomized schemes.

\begin{lemma}\label{hahe}
\hspace{-0.05cm}$R^*=\Rev(\calF)$ when $\calF_1,\ldots,\calF_n$ are i.i.d. with support size $2$ and 
  $k$ satisfies \emph{Lemma \ref{lem:mon}}.
\end{lemma}

We start with some preparation for the proof of Lemma \ref{hahe}.
%We first show some useful properties of these binomial probabilities.
%In Section \ref{sec:iid} we study $\rev(\calF)$ for the special case when $\calF_1,\ldots,\calF_n$
%  are i.i.d. with support size $2$.
%Below we let $\{a.b\}$ be the support with $a<b$ and let $p\in (0,1)$ be the probability of an item
%  taking the high value $b$.  
First recall that when  distributions are i.i.d.,  Daskalakis and Weinberg \cite{symmetry} showed that
  there always  exists an optimal solution to the standard LP for $\rev(\calF)$ (we use the
  price version in this section)  that is ``symmetric'': For any permutation $\sigma$ over  
  $[n]$ with $\sigma(\vv)=\ww$ (i.e.  $v_{\sigma(i)}=w_i$ for all $i\in [n]$),
  we always have $\sigma(\xx_\vv)=\xx_\ww$ and $\pi_\vv=\pi_\ww$, i.e., the lotteries bought at $\vv$ and $\ww$
  are the same under the permutation $\sigma$.
Based on that, one can significantly reduce the number of variables for the i.i.d. support-size $2$ case,
  and we refer to the new LP described below as the symmetric LP for $\rev(\calF)$.

%Recall that this LP has the following variables:
The symmetric LP has $3n+1$ variables:   $x_i$, for $i=1,\ldots,n$, is the probability~of getting an item with
value $b$ in $\vv$ when the valuation $\vv$ has $i$ items at $b$ (and $n-i$ items at 1); 
  $y_i$, for~$i=0,1,\ldots, n-1$,
  is the probability of getting an item with value $1$ when the valuation has $i$ items at $b$; 
  finally, $\pi_i$ for $i=0,1,\ldots,n$ is the price of the lottery for a valuation
with $i$ items at $b$.
%We also introduce two dummy variables $x_0=y_n=0$ that never appear in the LP but 
%  help simplify the presentation of some of the constraints.
The symmetric LP maximizes the expected revenue:
$
\sum_{i=0}^n \pi_i\cdot \calP_i
$
subject to the same constraints of the standard LP after replacing $\pi_\vv$ by $\pi_\ell$
  when $\vv$ has $\ell$ items at $b$ and $x_{\vv,i}$ by $x_\ell$ if $v_i=b$ and by $y_\ell$ if $v_i=1$. 
It is not hard to see that the number of distinct constraints left after the replacement is
  polynomial in $n$ and thus, the symmetric LP can be solved exactly in polynomial time.
By \cite{symmetry}, the optimal value of the symmetric LP is $\rev(\calF)$.
% we introduced to simplify the
%  presentation of the symmetric LP.
%We will relax the symmetric LP for optimal lotteries and show that
%the optimal value of even the relaxed LP is no greater than
%  $R^*$.

%Finally we introduce a simpler LP that captures $\rev(\calF)$ for the special
%  case when $\calF_1,\ldots,\calF_n$ are i.i.d. and have support size $2$.
%In particular, the number of variables is linear in $n$ instead of exponential in $n$.

We are now ready to prove Lemma \ref{hahe}.

\begin{proof}[Proof of Lemma \ref{hahe}]
Since $R^*$ is the  expected revenue of $S^*$,
  it suffices to show $\Rev(\calF)\le R^*$.
  
%To this end, we consider the symmetric Linear Program (see Section \ref{sec:preliminaries}) for the optimal menu of lotteries of the given~i.i.d. instance $\calF$. 
%Recall that this LP has the following variables:
%$x_i$ for $i=1,\ldots,n$ is the probability~of getting an item with
%value $b$ when the valuation has $i$ items at $b$ (and $n-i$ items at 1); 
%$y_i$ for~$i=0,1,\ldots, n-1$
%is the probability of getting an item with value 1 when the valuation has $i$ items at $b$; $\pi_i$ for %$i=0,1,\ldots,n$ is the price of the lottery for a valuation
%with $i$ items at $b$.
%Also recall the two dummy variables $x_0=y_n=0$ we introduced to simplify the
%  presentation of the symmetric LP.
For this purpose we will relax the
  symmetric LP for $\rev(\calF)$ and show that its optimal value is at most $R^*$. 
%%For notational convenience, we also introduce two dummy variables 
%  $x_0 = y_n =0$.\footnote{We introduce
%  $x_0$ and $y_n$ to simplify the constraints in (\ref{haha}) for the special cases when
%  $i=1$ and $n$, but note that the coefficient of $x_0$ (or $y_n$) in case $i=1$ (or $i=n$) is $0$
%  so they actually never appear in the relaxed LP.} 
In the relaxed LP we only keep the following constraints of the symmetric LP: %besides the nonnegativity of $x_i$ and $y_i$
\begin{flushleft}\begin{enumerate}
\item
  $0 \le x_i\le 1$ for each $i\in [n]$ and $0 \le y_i \leq 1$ for each $i\in [0:n-1]$.
\item $\pi_0 \leq n y_0$ (i.e., the utility at the all-$1$ vector is nonnegative);
\item For each $i\in [n]$, the constraint that the
valuation $\ww$ with $w_j=b$ for $j\in [i]$ and $w_j=1$ for\\ $j>i$ does not
envy the  lottery of $\vv$ with $v_j=b$ for $j\in [i-1]$ and $v_j=1$ for $j>i-1$:
\begin{equation}\label{haha}
b i  x_i  +(n-i)y_i - \pi_i \geq b (i-1) x_{i-1}  +(n-i+b)y_{i-1} - \pi_{i-1}.
\end{equation}
Note that when $i=1$, $x_0$ appears on the RHS with coefficient $0$;
  when $i=n$, $y_n$ appears on the LHS with coefficient $0$. 
For convenience we introduce $x_0=y_n=0$ as dummy variables that never appear in the relaxed LP but   
  help simplify the presentation of these constraints.
\end{enumerate}\end{flushleft}
Since all of them are part of the symmetric LP,
  the optimal value of the relaxed LP is at least $\rev(\calF)$.
In the rest of the proof we show that the optimal value of the relaxed LP is at most $R^*$.  

To this end we use the constraints above to upperbound each $\pi_i$ using $x$ and $y$ variables.
For $i=0$ we use $\pi_0 \leq ny_0$.
For each $j\in [n]$ we have from constraints (2) in the relaxed LP that:
$$\pi_j \leq \pi_{j-1} +b j x_j  +(n-j)y_j - b(j-1)x_{j-1} - (n-j+b)y_{j-1}.$$
Summing these inequalities for all $j=1,\ldots,i$, we get after some cancellations:
$$\pi_i \leq \pi_0 +b i x_i  +(n-i)y_i - (b-1)(y_{i-1}+y_{i-2}+\cdots +y_1)-(n+b-1)y_0.$$
Plugging in $\pi_0 \leq ny_0$, we have for each $i\in [n]$:
$$\pi_i \leq bi x_i  +(n-i)y_i - (b-1)(y_{i-1}+y_{i-2}+\cdots + y_1+y_0).$$

Replacing in the objective function $\sum_i \calP_i \cdot \pi_i$ each $\pi_i$ by its upper bound, we get a linear form in the $x_i$'s, $i\in [n]$, and $y_i$'s, $i\in [0:n-1]$, which upperbounds the
value of the relaxed LP (note that $x_0$ and $y_n$ are dummy variables that do not really appear
  in any constraint).
For each $i \in [n]$, the coefficient of $x_i$ in the linear form is $bi\cdot \calP_i$, thus
this term is maximized if we set $x_i=1$ for each $i\in [n]$. %(note that the
%  coefficient of the dummy variable $x_0$ is $0$).
The coefficient of $y_0$ is $n \calP_0 - (b-1)(\calP_1 + \cdots + \calP_n)$ 
and the coefficient of $y_i$ for each $i\in [n-1]$ is 
$(n-i) \calP_i - (b-1)(\calP_{i+1}+ \cdots + \calP_n)$.
From the choice of $k\in [0:n]$ and Lemma \ref{lem:mon},
  we have for all $i\in [0:n-1]$: The coefficient of $y_i$ is negative if $i<k$, and 
  is nonnegative if $i\ge k$.  
%definition of the index $k$ as 
%the minimum index for which $(n-k) \calP_k \geq (b-1)\sum_{j\geq k+1} \calP_j$.
%Since $\frac{(n-i) \calP_i}{\sum_{j\geq i+1} \calP_j}$ is monotone increasing by Lemma \ref{lem:mon}, for $i<k$ %the ratio is $< (b-1)$, hence the
%coefficient of $y_i$ is negative, while for $i \geq k$ the ratio is $\geq (b-1)$,
%hence the coefficient of $y_i$ is positive.
Therefore, the linear form is maximized when we set $y_i=0$ for all $i<k$ 
and $y_i=1$ for all $i \geq k$.
Applying these substitutions in the linear form, the upper bound
on the value of the LP becomes
(note that $(n-i) \calP_i - (b-1)(\calP_{i+1}+ \cdots + \calP_n)$ is $0$ when $i=n$):
\begin{align*}
&\sum_{1\le i\le n} bi\cdot \calP_i + \sum_{k\le i\le n-1} \big[(n-i) \calP_i - (b-1)(\calP_{i+1}+ \cdots + \calP_n)\big]\\[0.4ex]
&=\sum_{1\le i\le n} bi\cdot \calP_i + \sum_{k\le i\le n} \big[(n-i) \calP_i - (b-1)(\calP_{i+1}+ \cdots + \calP_n)\big]\\[0.4ex]
&=\sum_{1\le i\le n} bi\cdot\calP_i + \sum_{k\le i\le n} \calP_i\cdot \big[(n-i) - (b-1)(i-k)\big]\\[0.4ex]
&=\sum_{1\le i<k} bi\cdot \calP_i + \sum_{k\le i\le n} \calP_i\cdot  [n+(b-1)k] 
= R^*.
\end{align*}
This finishes the proof of the lemma.
%which after rearrangements is exactly the expected revenue $Rev(S^*)$ of $S^*$.
\end{proof}

{It remains to prove Lemma \ref{lem:mon}}. We start with the following lemma.

\begin{lemma} \label{lem:ineq}
For all $i=1,\ldots,n-1$, $\calP_i (n-i) + \sum_{j\geq i+1} \calP_j \left(n-\frac{i}{p}\right) > 0$.
\end{lemma}
\begin{proof}
We use induction on $n-i$.\\
{\em Basis}: $n-i=1$, i.e. $i=n-1$. The left-hand side is
$$\calP_{n-1} + \calP_n \left(n-\frac{n-1}{p}\right)=n \cdot p^{n-1}\cdot(1-p)+p^n \cdot 
\left(n-\frac{n-1}{p}\right)=  p^{n-1} >0.$$
{\em Induction Step}: 
We have $\calP_i (n-i) =  \calP_{i+1} \cdot (i+1) \cdot \frac{1-p}{p}$.
Therefore, the left-hand side is equal to
\begin{align*}
LHS & =  \calP_{i+1} \cdot (i+1) \cdot \frac{1-p}{p} + \calP_{i+1} \cdot \left(n-\frac{i}{p}\right)
+ \sum_{j\geq i+2} \calP_j \left(n-\frac{i}{p}\right)\\[0.4ex]
& = \calP_{i+1} (n-(i+1)) + \calP_{i+1}\cdot \frac{1}{p} + \sum_{j\geq i+2} \calP_j \left(n-\frac{i}{p}\right)\\[0.5ex]
& > \calP_{i+1} (n-(i+1)) +  \sum_{j\geq i+2} \calP_j \left(n-\frac{i+1}{p}\right) > 0,
\end{align*}
where the last inequality holds by the induction hypothesis.
 \end{proof}
 
Now we are ready to prove Lemma \ref{lem:mon}: 
\begin{proof}[Proof of Lemma \ref{lem:mon}]
We let $k$ be the smallest $i\in [0:n]$ such that (\ref{eq:jaja}) is nonnegative
  ($k$ is well defined as (\ref{eq:jaja}) is $0$ when $i=n$).
To prove that (\ref{eq:jaja}) is nonnegative for all $i\ge k$, it suffices to show that 
$$\frac{(n-i) \calP_i}{\sum_{j\geq i+1} \calP_j}$$
is monotonically increasing for $i$ from $0$ to $n-1$.
Fix an $i\in [n-1]$. Our goal is to show that
$$\frac{(n-i) \calP_i}{\sum_{j\geq i+1} \calP_j} - \frac{(n-(i-1))\calP_{i-1}}{\sum_{j\geq i} \calP_j} > 0,$$
or equivalently, $(n-i) \calP_i \sum_{j\geq i} \calP_j - (n-(i-1))\calP_{i-1} \sum_{j\geq i+1} \calP_j >0$.
Since $$(n-(i-1))\calP_{i-1} = i\calP_i\cdot \frac{1-p}{p},$$ we can rewrite the
left-hand side as
$$\calP_i \left[ (n-i) \sum_{j\geq i} \calP_j - i \frac{1-p}{p} \sum_{j\geq i+1} \calP_j\right]
=\calP_i\left[
(n-i) \calP_i + \sum_{j\geq i+1} \calP_j (n - (i/p))\right],$$
which is positive by Lemma \ref{lem:ineq}.
Therefore, the desired inequality holds.
\end{proof}

%\end{flushleft}
%\end{document}

% !TEX root = main.tex

\section{Hardness of Revenue-Optimal Deterministic Mechanism Design}\label{sec:hard}

%\documentclass[11pt,letterpaper]{article}
%\usepackage{hyperref}

%\usepackage{lmodern}
%\usepackage[T1]{fontenc}
%\usepackage[utf8]{inputenc}
%\usepackage{ragged2e}

%\usepackage{verbatim}

%\usepackage{multirow}
%\usepackage{url}
%\usepackage{setspace}
%\usepackage{times}
%\usepackage{fullpage}
%\usepackage{epsfig}
%\usepackage[usenames]{color}
%\usepackage{graphicx}
%\usepackage{times}
%\usepackage{enumerate}

%\newtheorem{theorem}{Theorem}[section]
%\newtheorem{lemma}[theorem]{Lemma}
%\newtheorem{property}[theorem]{Property}
%\newtheorem{proposition}[theorem]{Proposition}
%\newtheorem{fact}[theorem]{Fact}
%\newtheorem{corollary}[theorem]{Corollary}
%\newtheorem{conjecture}[theorem]{Conjecture}
%\newtheorem{claim}[theorem]{Claim}
%\newtheorem{observation}[theorem]{Observation}
%\newtheorem{definition}[theorem]{Definition}
%\newtheorem{remark}[theorem]{Remark}
%\newtheorem{construction}[theorem]{Construction}
%\newtheorem{mechanism}{Mechanism}

%%\def\qsym{\vrule width0.6ex height1em depth0ex}
%%\newcount\proofqeded
%%\newcount\proofended
%%\def\qed{ \mbox{\ \vrule width1ex height1em depth0cm}
%%\global\advance\proofqeded by 1 }
%%\newenvironment{proof}{\proofstart}{\ifnum\proofqeded=\proofended\qed\fi \global\advance\proofended by 1
%%  \medskip}
%%\makeatletter
%%\def\proofstart{\@ifnextchar[{\@oprf}{\@nprf}}
%%\def\@oprf[#1]{\paragraph{Proof #1:~}}
%%\def\@nprf{\paragraph{Proof:~}}
%%\makeatother

%\usepackage[left=1in,right=1in,top=0.9in, bottom=0.92in]{geometry}

%\newcommand{\Exp}{\operatornamewithlimits{\mathbb{E}}}
%\newcommand{\cube}{\operatorname{\{0, 1\}}}
%\newcommand{\pcube}{\operatorname{\{\pm1\}}}
\newcommand{\eps}{\epsilon}

%\newcommand{\F}{\mathbb{F}}
%\newcommand{\C}{\mathbb{C}}
%\newcommand{\R}{\mathbb{R}}
%\newcommand{\Z}{{\mathbb Z}}
%\newcommand{\N}{{\mathbb N}}
%\newcommand{\calf}{{\cal F}}
%\%newcommand{\calm}{{\cal M}}
%\newcommand{\eqdef}{{\stackrel{\rm def}{=}}}
%\newcommand{\Fp}{{\mathbb F}_{p}}
%\newcommand{\polylog}[1]{\mathrm{polylog}(#1)}
%\newcommand{\poly}[1]{\mathrm{poly}(#1)}
%\newcommand{\diff}[1]{\mathrm{diff}(#1)}
%\newcommand{\wt}[1]{\mathrm{wt}(#1)}
%\newcommand{\ith}[1]{{#1}^{\text{th}}}

%\newcommand{\pairs}{\mathrm{pairs}}
%\newcommand{\social}{\mathrm{social}}
%\newcommand{\ud}{\mathrm{d}}
%\newcommand{\pr}{\mathrm{P}}
%\newcommand{\mmax}{\mathrm{Max}}

%\newcommand{\dominates}{\preccurlyeq}
%\newcommand{\dominated}{\succcurlyeq}

\def\rr{\mathbf{r}} \def\calC{\mathcal{C}} \def\xx{\mathbf{x}} \def\11{\mathbf{1}}
\def\pr{\textsf{Pr}} \def\vv{\mathbf{v}} \def\uu{\mathbf{u}} \def\cc{\mathsf{CC}}
\def\xx{\mathbf{x}} \def\yy{\mathbf{y}} \def\lpr{\mathsf{L}} \def\rpr{\mathsf{R}}
\def\ext{\mathsf{Ext}} \def\aveL{\mathsf{aveL}} \def\aveR{\mathsf{aveR}}
\def\minL{\mathsf{minL}} \def\minR{\mathsf{minR}} \def\ss{\mathbf{s}}
\def\den{\mathsf{den}} \def\upper{\mathsf{upper}} \def\low{\mathsf{lower}}
\def\zz{\mathbf{z}} \def\denL{\mathsf{denL}} \def\denR{\mathsf{denR}}
\def\calP{\mathcal{P}} \def\pp{\mathbf{p}} \def\qq{\mathbf{q}} \def\ww{\mathbf{w}}
\def\calM{\mathcal{M}} \def\calE{\mathcal{E}} \def\calD{\mathcal{D}}
\def\supp{\textsc{supp}} \def\brev{\textsc{BRev}} \def\calF{\mathcal{F}}
\def\brevh{\textsc{BRev}_\text{H}} \def\calN{\mathcal{N}}
\def\brevl{\textsc{BRev}_\text{L}} \def\calT{\mathcal{T}} \def\Rev{\mathsf{Rev}}
\def\rev{\mathsf{Rev}} \def\calS{\mathcal{S}} \def\calT{\mathcal{T}}
\def\calL{\mathcal{D}} \def\calU{\mathcal{U}} \def\calQ{\mathcal{Q}}

%\spacing{1.25}
We prove Theorem \ref{thm:main} in this section.
The plan is to 
%\begin{document}
%\begin{flushleft}\setlength{\parindent}{15pt}
reduce from the following \#P-hard decision problem called \textsc{COMP} introduced in \cite{CD+15}. 
The input consists of three parts: 1) a set $B$ of $n$ nonnegative 
  integers $B=\{b_1,\ldots,b_n\}$ between $0$ and $2^n$ (we assume
  without loss of generality that $b_1\le \cdots\le b_n$);
  2) a subset $W\subset [n]$ of size $|W|=n/2$ (assume without loss of generality that $n$ is even) and we use 
   $w$ to denote $\sum_{i\in W} b_i$; and 3) an integer $t$.
The question is to decide whether the number of $S\subset [n]$
  such that $|S|=n/2$ and $\sum_{i\in S} b_i\ge w$ is at least $t$ or at most $t-1$.
While COMP was shown to be \#P-hard in \cite{CD+15},
  we need here the same problem with the following two extra conditions on the two input sets $B$ and $W$, which 
  we will refer to as COMP$^*$:
\begin{flushleft}\begin{enumerate}
\item Every $(n/2)$-subset $S\subset [n]$ with $b_n\in S$ has
  $\sum_{i\in S} b_i\ge w$, i.e. $b_1+\cdots+b_{n/2-1}+b_n\ge w$.
\item Every $(n/2)$-subset $S\subset [n]$ that does not contain $b_n$ but 
  contains either $b_1$ or $b_2$ must have $\sum_{i\in S} b_i<w$, i.e. $b_2+b_{n/2+1}+\cdots+b_{n-1}<w$.
\end{enumerate}\end{flushleft}
These extra conditions will come in handy in the reduction below. 

\subsection{\#P-Hardness of \texorpdfstring{COMP$^*$}{COMP*}}
We show that COMP$^*$ is \#P-hard  %in Appendix \ref{appendix1} (with a simple reduction from the original COMP). 
via a polynomial-time reduction from COMP.
Let $(B,W,t)$ be an input to the original COMP problem, where $B=\{b_1,\ldots,b_n\}$
  with $b_1,\ldots,b_n\in [0:2^n]$, $W\subset [n]$ with $|W|=n/2$, and $w=\sum_{i\in W} b_i$.
We first construct from $B$ a new set $B'$ of $4n$ integers between $0$ and $2^{4n}$:
\begin{enumerate}
\item $B'$ contains $n$ integers $2^{2n}+b_i$, $i\in [n]$;
\item $B'$ contains $2^{4n}$ (as the largest integer in $B'$) and $3n/2$ copies of $2^{3n}$;
\item The rest of $3n/2-1$ integers of $B'$ are all $0$ (so the smallest two integers in $B'$ are $0$).
\end{enumerate}
The new set $W'$ contains $2n$ integers from $B'$: the $3n/2$ copies of $2^{3n}$
  and each of $2^{2n}+b_i$, $i\in W$. Let 
$$
w'=\sum_{b\in W'} b=(3n/2)\cdot 2^{3n}+ (n/2)\cdot 2^{2n}+w\quad\text{and}\quad
t'=\binom{4n-1}{2n-1} +t.
$$

We show that $B'$ and $W'$ satisfy the two extra conditions.
For the first part note that $\sum_{b\in S} b>w'$ for any subset $S$ that contains the
  largest integer $2^{4n}$.
On the other hand, if $S$ is a $2n$-subset of $W'$ that does not contain $2^{4n}$ but contains 
  at least one $0$, then we have (using $b_i\le 2^n$ for all $i\in [n]$)
$$
\sum_{b\in S} b\le (3n/2)\cdot 2^{3n}+(n/2-1)\cdot (2^{2n}+2^n)<w'.
$$ 

Let $S$ be a $2n$-subset that satisfies $\sum_{b\in S} b\ge w'$.
Then either i) $S$ contains $2^{4n}$, or ii) $S$ contains all the $3n/2$ copies of $2^{3n}$
  and a subset of $n/2$ integers $2^{2n}+b_i$, with the sum of the latter 
  being at least $(n/2)\cdot 2^{2n}+w$.
This then implies that $(B,W,t)$ is a yes-instance if and only if $(B',W',t')$ is a yes-instance.
It follows that COMP$^*$ is also \#P-hard.

\subsection{The Reduction}

We now present the reduction from COMP$^*$ to the optimal bundle-pricing problem.
%\edit{revenue-optimal deterministic mechanism design, which we will refer to as bundle-pricing for brevity}.

Given an input instance $(B,W,t)$ of COMP$^*$ (with $B$ and $W$ satisfying the extra conditions) 
  we define an instance $\calF$ of optimal bundle-pricing with $n+1$ items
  and support size $2$.
%  (each with two possible values).
%Below we assume without loss of generality that $n$ is even. 
We shall refer to the first $n$ items as item $i$ for $i\in [n]$ and 
  refer to the last item as the \emph{special} item. 
Let 
$$h=2^{2n},\quad p=\frac{1}{2(h+1)},\quad 
  {\delta=\frac{1}{2^{3n}}},\quad a_i=b_i\delta \quad\text{and}\quad h_i=h+a_i,\quad\text{for each $i\in [n]$} 
$$
(hence $a_i\in [0,2^n\delta]$ and is
  an integer multiple of $\delta$).
Then item $i$ is supported on $\{1,h_i+1\}$. The probability of $h_i+1$
  is $p$ and the probability of $1$ is $1-p$.
Let 
$$
c=w\delta,\quad \alpha = (n/2)h +c,\quad\sigma=\frac{1}{p^{n}}\quad\text{and}\quad
\tau = \frac{\sigma}{\sigma+\alpha}
$$ 
(hence $c\in [0,(n/2)2^n\delta]$ and is an integer multiple of $\delta$).  %and $$. 
%Let $\color{red}$. % be a parameter that is exponential in $n$ %(in particular,
  %we require that $\color{red}\sigma\ge 2^{3n}$ for now) 
%  but will be specified later.
The special item is supported on $\{\sigma, \sigma+\alpha \}$.
The probability of $\sigma+\alpha $ is $\tau-\eps$  
and the probability of $\sigma$ is $1-\tau+\eps= ({\alpha}/({\sigma+\alpha})) +\eps$ for some $\eps$ (which 
  is not necessarily positive) to be specified at the end of the proof; for now we only 
  require that ${|\eps|=o(1/\sigma)}$.
Note that since $\sigma \gg \alpha$, 
  the probability $\tau-\eps$ of $\sigma+\alpha$ is very close to 1,
and the probability $1-\tau+\eps$ of $\sigma$ is positive but very close to $0$.
This finishes the description of the bundle pricing instance $\calF$ (except the choice of 
  the parameter $\eps$ which we will set at the end). % (it is exponentially small).
%The two grand bundle prices we are interested in are $\sigma+n$ and $2\sigma+n$, respectively.

\subsection{Plan of the Proof}

Our plan for the proof is the following. 
In Section \ref{sec:nota}, we introduce some notation and define two
  simple bundle-pricings (Solution 1 and 2), as feasible solutions to the standard IP for $\drev(\calF)$,
  both of which are discounted item-pricings.
Most of the work lies in Section \ref{juju}, where we show that whenever $|\eps|=o(1/\sigma)$,
  one of these two solutions is the \emph{unique} optimal solution to the standard IP and achieves $\drev(\calF)$.
This is done by relaxing the standard IP  and showing that 
  one of these two solutions is the unique optimal solution to the relaxed IP.
As they are both feasible to the standard IP,
  we conclude that one of them is uniquely optimal for the standard IP.
%Indeed we prove a stronger statement (which is necessary for the proof of Theorem \ref{thm:main}): % is stronger:
%\begin{flushleft}\begin{quote}
%Let $\vv^*$ denote the valuation where every item takes its low value.  
%We show that given any feasible solution to the relaxed IP,
%  if the special item is included in the bundle bought at $\vv^*$,
%  then its expected revenue is at most that of Solution 1;
%%  otherwise, its expected revenue is at most that of Solution 2.
%\end{quote}\end{flushleft}

Finally, we set $\eps$ carefully (with $|\eps|=o(1/\sigma)$ as promised) in Section \ref{finish} to show that
  Solution 2 is strictly better if the $(B,W,t)$ used in the construction of $\calF$
  is a yes-instance of problem COMP$^*$,
  and Solution 2 is strictly better if it is a no-instance.
%  As a result, the special item must be included in the bundle at $\vv^*$
%  in any optimal menu of bundles if $(B,W,t)$ is a yes-instance,
%  and must be excluded in any optimal menu if $(B,W,t)$ is a no-instance.
This finishes the proof of Theorem \ref{thm:main}.
  
\subsection{Notation and two simple solutions}\label{sec:nota}

%$We introduce some notation.
For convenience, we will use a subset $S\subseteq [n]$ to denote a valuation vector
  over the first $n$ items, where $i\in S$ (or $i\notin S$) means that item $i$ takes the high  value
  $h_i+1$ (or low value $1$).
We will also use $(S,\sigma+\alpha)$ (or $(S,\sigma)$) to denote a full valuation vector over all the $n+1$
  items in which the special item has the high value $\sigma+\alpha$ (or low value $\sigma$).  

Given $S\subseteq [n]$, we write $\Pr[S]$ to denote
  $p^{|S|}(1-p)^{n-|S|}$; given an integer $i\in [0:n]$, we write
  $$\Pr[i]=\binom{n}{i}\cdot p^i(1-p)^{n-i},$$
  use $\Pr[i\ge k]$ to denote $\sum_{i=k}^n \Pr[i]$,
  and $\Pr[i>k]$ to denote $\sum_{i=k+1}^n \Pr[i]$.
%We use $(S,\sigma)$ to denote the configuration where each item $i\in S$
%  takes the high value $h_i+1$, each item $i\notin S$ takes the low value $1$, and
%  the special item takes the low value $\sigma$;
%  we use $(S,\sigma+\alpha )$ to denote the configuration where each item $i\in S$
%  takes the high value $h_i+1$, each item $i\notin S$ takes the low value $1$,
%  and the special item takes the high value $\sigma+\alpha $.
We also use $\Pr[S,\sigma]$ and $\Pr[S,\sigma+\alpha]$ to denote the probabilities of valuations
  $(S,\sigma)$ and $(S,\sigma+\alpha)$: 
$$\Pr[S,\sigma]=\Pr[S]\cdot (1-\tau+\eps)\quad\text{and}\quad
\Pr[S,\sigma+\alpha ]=\Pr[S]\cdot (\tau-\eps).$$
  %with $\Pr[S]=p^{|S|}(1-p)^{n-|S|}$.

We use the standard IP for $\drev(\calF)$ (the utility version) but rename the variables as follows.
For each $S\subseteq [n]$ we use $x_{S,i}\in \{0,1\}$ to denote the 
  variable for item $i$ in valuation $(S,\sigma)$ for each $i\in [n]$,
  $z_{S}\in \{0,1\}$ to denote the variable for the special item,
  and $u_S\ge 0$ to denote the utility.
For each $S\subseteq [n]$, we use $x_{S,i}'\in \{0,1\}$ to denote
  the variable for item $i$ in valuation $(S,\sigma+\alpha )$, $z_{S}'\in \{0,1\}$
  to denote the variable for the special item, and $u_S'\ge 0$ to denote the utility.
%Configuration $(i,1)$, $i\in [0:n]$: there are (exactly) $i$ items at the high value $h+1$
%(the remaining $n-i$ items are at the low value 1), and the 
%  special item is at the low value $\sigma$; configuration $(i,2)$, $i\in [0:n]$: there
%  are $i$ items at the high value $h+1$ and the special item is at the high value $\sigma+\alpha$.
%For convenience,
%we use $(x_{i},y_i,z_i)$ to denote the lottery of configuration $(i,1)$
%  and $(x_{i}',y_{i}',z_{i}')$
%  to denote the lottery of configuration $(i,2)$, where $x_{i}$ and $x_{i}'$
%  are the probability of getting an item at $h+1$, $y_{i}$ and $y_{i}'$ are 
%  the probability of getting an item at $1$, $z_{i}$ and $z_{i}'$ are 
%  the probability of getting the special item.
%Here we consider deterministic solutions so all the variables are $\{0,1\}$ variables.
%We also use 
%  $u_{S}$ and $u_{S}'$ to denote the utilities, which are nonnegative real variables.
%We use $\xx, \zz,\xx', \zz'$ and $\uu,\uu'$ to denote the vectors of variables.

The two simple bundle-pricings we are interested in are the following:
\begin{flushleft}\begin{enumerate}
\item[] {\em \hspace{-0.5cm}Solution 1}: Offer the grand bundle at $\sigma+n$, or (equivalently) offer
each item at its low value. 
\item[] {\em \hspace{-0.5cm}Solution 2}: The discounted item-pricing where the grand bundle is
  offered at $\sigma+\alpha+n$ and each individual item is offered at its high value.
As discussed in the introduction, this means that the buyer can buy either the grand bundle at
  $\sigma+\alpha+n$ or any bundle of items at the sum of their high values.
When this menu is offered, the buyer buys the grand bundle if its utility is positive
  (since her utility from buying items priced at their high values can never be positive)
  and buys the bundle of items at their high values (if any) if the utility from the grand bundle is negative.
For the case when the utility from the grand bundle is $0$, the buyer still gets the grand bundle
  since it always gives a higher revenue.
%The latter means that we offer, for each $S\subseteq [n]$, the bundle of items $i\in S$
%  at $\sum_{i\in S} (h_i+1)$.
\end{enumerate}\end{flushleft}
These two bundle-pricings induce two feasible solutions to the standard IP:
\begin{flushleft}\begin{enumerate}
\item In Solution 1, every valuation $(S,\sigma)$ and $(S,\sigma+\alpha)$ buys the grand bundle.
So we have for each $S\subseteq [n]$ and $i\in [n]$,
  $x_{S,i}=z_{S}=x_{S,i}'=z_{S}'=1$. 
Regarding the utilities we have 
\begin{eqnarray*}
&u_S=\sum_{i\in S} h_i\quad\text{and}\quad
u_S'=\alpha +\sum_{i\in S} h_i.&
\end{eqnarray*}
\item In Solution 2 we have 1) for each $S\subseteq [n]$ and $i\in [n]$, $x_{S,i}'=z_{S}'=1$ and 
  $u_S'= \sum_{i\in S} h_i$; 2) for each $S\subseteq [n]$ with $\sum_{i\in S} h_i\ge \alpha $, 
  we have $x_{S,i}=z_{S}=1$ for all $i\in [n]$, $u_{S}=\sum_{i\in S}h_i-\alpha;$  3) for each $S\subseteq [n]$ with 
  $\sum_{i\in S} h_i<\alpha$, we have 
  $x_{S,i}=1$ for each $i\in S$, $x_{S,i}=0$ for each $i\notin S$, $z_S=0$, and  $u_S=0$.
Given our choice of parameters (i.e. $h\gg a_i,c$), every $S$ with $|S|>n/2$ satisfies case 2) and every 
  $S$ with $|S|<n/2$ satisfies case 3).
A set $S$ with $|S|=n/2$ satisfies case 2) if we have 
$\sum_{i\in S} a_i\ge c$ (equivalently, $\sum_{i\in S} b_i\ge w$)  
  and satisfies case 3) otherwise.
This is the connection with COMP$^*$ that we will explore in the reduction.
%Note that $\sum_{i\in S}a_i\ge c$ iff $\sum_{i\in S} b_i\ge w$.
\end{enumerate}\end{flushleft}  

%{\color{red}Below we use $h_i$ to denote $h+a_i$.}

As discussed in the plan, we introduce a relaxation of the standard IP that only contains a subset
  of its constraints and refer to it as the \emph{relaxed} IP.
It contains the following constraints:
\begin{enumerate}
\item $x_{S,i},z_S,x_{S,i}',z_S'\in \{0,1\}$ and $u_S,u_S'\ge 0$, for all $S\subseteq [n]$ and $i\in [n]$;
\item For each $S\ne \emptyset$, $(S,\sigma+\alpha)$ does not envy
  $(\emptyset,\sigma+\alpha)$:
\begin{eqnarray*}
&u'_S\ge u'_{\emptyset}+\sum_{i\in S} h_i\cdot x_{\emptyset,i}'.&
\end{eqnarray*}
%$$
%u'_S\ge u'_{\emptyset}+\sum_{i\in S} h_i\cdot x_{\emptyset,i}'.
%$$
\item For each $S\subseteq [n]$, $(\emptyset,\sigma+\alpha)$ does not envy
  $(S,\sigma)$:
\begin{eqnarray*}
&u_\emptyset'\ge u_S-\sum_{i\in S} h_i\cdot x_{S,i}+\alpha\cdot z_{S}.&
\end{eqnarray*}
%$$
%u_\emptyset'\ge u_S-\sum_{i\in S} h_i\cdot x_{S,i}+\alpha\cdot z_{S}.
%$$
\item For each $S\subseteq [n]$, 
  $(S,\sigma)$ does not envy $(\emptyset,\sigma+\alpha)$:
\begin{eqnarray*}
&u_S\ge u_\emptyset'+\sum_{i\in S} h_i\cdot x_{\emptyset,i}'-\alpha\cdot z_\emptyset'.&
\end{eqnarray*}
%$$
%u_S\ge u_\emptyset'+\sum_{i\in S} h_i\cdot x_{\emptyset,i}'-\alpha\cdot z_\emptyset'.
%$$
\item For each pair of $T,S$ with $T\subset S\subseteq [n]$, $(S,\sigma)$ 
  does not envy $(T,\sigma)$:
\begin{eqnarray*}
&u_S\ge u_T+\sum_{i\in S\setminus T} h_i\cdot x_{T,i}.&
\end{eqnarray*}
%$$
%u_S\ge u_T+\sum_{i\in S\setminus T} h_i\cdot x_{T,i}.
%$$
\end{enumerate}
The objective function of the relaxed IP is the same expected revenue:
\begin{align*}
&\sum_{S\subseteq [n]} \left(\sum_{i\in S} (h_i+1)\cdot x_{S,i}+\sum_{i\notin S} x_{S,i}+
  \sigma\cdot z_S-u_S\right) \cdot \Pr[S,\sigma]\\
&\hspace{1cm}+ \sum_{S\subseteq [n]} \left(\sum_{i\in S} (h_i+1)\cdot x_{S,i}'+\sum_{i\notin S} x_{S,i}'
   +(\sigma+\alpha)\cdot z_S'-u_S'\right)\cdot  \Pr[S,\sigma+\alpha].
\end{align*}

\subsection{One of the two simple solutions is optimal}\label{juju}

Our goal is the following lemma about optimal solutions to the relaxed IP:

\begin{lemma}\label{lem:main}
In an optimal solution to the relaxed IP, if $z_\emptyset=1$, then it must be Solution 1;
  if $z_\emptyset=0$, then it must be Solution 2.
%One of the two grand bundle solutions described above is optimal for the relaxed IP.
\end{lemma}

As both solutions are feasible for the standard IP, we have the following corollary:

\begin{corollary}
Either Solution 1 or 2 is the unique optimal solution to the standard IP. % and achieves $\emph{\drev}(\calF)$.
\end{corollary}

We start the proof of Lemma \ref{lem:main} with a few simple observations.

\begin{lemma}\label{lem:1}
In any optimal solution to the relaxed IP, we have
  $x_{S,i}=1$ for all $S\subseteq [n]$ and $i\in S$,
  $x_{S,i}'=1$ for all $S\ne \emptyset$ and $i\in [n]$, and
  $z_S'=1$ for all $S\subseteq [n]$.
%We also have $u_\emptyset=0$.
  %$$u_i'=\min\Big\{0,\min_{j<i} \big\{u_j+(i-j)h\cdot y_j+\sigma\cdot z_j\big\},\min_{j\ge i} \big\{u_j-(j-i)h\cdot x_j+\sigma\cdot z_j\big\}\Big\}.$$
\end{lemma}
\begin{proof}
The part of $\smash{x_{S,i}'=1}$ for all $S\ne \emptyset$ and $i\in [n]$
  and $\smash{z_S'=1}$ for all $S\ne \emptyset$
  is trivial as they do not appear in the constraints of the relaxed IP but
  appear with a positive coefficient in the objective function.
For $z_\emptyset'$, note that it only appears on the 
  right hand side with a negative coefficient.
Thus, if $z_\emptyset'=0$ in a feasible solution, we can switch it to $1$ and the
  new solution remains feasible but the expected revenue goes up strictly.
The same argument works for $x_{S,i}$ for all $S\subseteq [n]$ and $i\in S$.
%This finishes the proof of the lemma.
%Finally we prove that $x_{\emptyset,i}'=1$ for all $i\in [n]$.
%If we have $x_{\emptyset,i}'=0$ in a feasible solution for some $i\in [n]$,
%  then we update the solution by setting $x_{\emptyset,i}'=1$ and adding $h$ to
%  each $u_S'$ with $i\in S$.
%It is easy to check that the new solution is still feasible and 
%  the expected revenue goes up by
%$$
%\Pr[\emptyset,\sigma+\alpha]-h\cdot \sum_{S:i\in S}\Pr[S,\sigma+\alpha]
%=(\tau-\eps)\cdot \left((1-p)^n-ph\right)>(\tau-\eps)(1-pn-ph)=0,
%$$
%by plugging in $p=1/(h+n)$.
%This finishes the proof of the lemma.
\end{proof}

Next we show that $x_{\emptyset,i}'=1$ for all $i\in [n]$ in any optimal solution to the relaxed IP.

\begin{lemma}\label{lem:3}
In an optimal solution to the relaxed IP, $x'_{\emptyset,i}=1$ for all $i\in [n]$ and
  $u_S'=u_\emptyset'+\sum_{i\in S} h_i$ for all nonempty $S\subseteq [n]$.
\end{lemma}
\begin{proof}
Assume for contradiction that $x'_{\emptyset,i}=0$ for $i\in I$ and $I$ is nonempty.
Then we make the following changes to obtain a new solution:
\begin{enumerate}
\item[i)] We change $x'_{\emptyset,i}$ from $0$ to $1$ for each $i\in I$;
\item[ii)] For each $S\ne \emptyset$, we increase $u_S'$ by $\sum_{i\in S\cap I} h_i$;
\item[iii)] For each $S\subseteq [n]$, we replace $u_S$ by the maximum of the original $u_S$
  and $u_\emptyset'+\sum_{i\in S} h_i-\alpha $\\ (by doing this the new $u_S$ can go up
  by at most $\sum_{i\in S\cap I} h_i$ because of constraints (4)).
\end{enumerate}
  
We first verify that the new solution remains feasible and then show that
  its expected revenue is strictly better than that of the original solution.
For its feasibility, constraints in (2) and (4) are trivial.
 (3) for $S$ holds trivially if $u_S$ remains the same.
Otherwise $u_S=u_\emptyset'+\sum_{i\in S}h_i-\alpha $ and thus,
\begin{eqnarray*}
&u_S-\sum_{i\in S} h_i\cdot x_{S,i}+\alpha\cdot z_S
=u_S-\sum_{i\in S} h_i+\alpha\cdot z_S
\le u_S-\sum_{i\in S} h_i+\alpha=u_\emptyset'.&
\end{eqnarray*}
For 5), the constraint for $T\subset S$ holds trivially if $u_T$ remains the same.
If $u_T$ goes up, we have
\begin{eqnarray*}
&u_S\ge u_\emptyset'+\sum_{i\in S}h_i-\alpha\quad\text{and}\quad
u_T=u_\emptyset'+\sum_{i\in T}h_i-\alpha&
\end{eqnarray*}
and thus, $u_S-u_T\ge \sum_{i\in S\setminus T} h_i\ge \sum_{i\in S\setminus T}h_i\cdot x_{T,i}.$
This shows that the new solution is also feasible.

Finally, comparing the two solutions, the net gain of expected revenue in
  the new one is at least
\begin{equation*}  
|I|\cdot \Pr[\emptyset,\sigma+\alpha] - \sum_{S\subseteq [n]} \big(\Pr[S,\sigma+\alpha]
  +\Pr[S,\sigma]\big)\left(\sum_{i\in S\cap I} h_i \right)= |I|\cdot 
  ( \tau-\epsilon)\cdot (1-p)^n-\sum_{i\in I} ph_i>0
\end{equation*}
as by our choice of parameters 
  the {first term is close to $|I|$ and the second term is close to $|I|/2$.}
  
The second part of the lemma follows trivially since $u_S'$ appears only in the LHS of constraints (2)
  when $S\ne \emptyset$ and appears in the objective function with a negative coefficient.   
\end{proof}

\def\sol{\mathsf{Sol}}

%Therefore, an optimal solution is uniquely determined by the two vectors $\xx$ and $\zz$
%  (with $x_{S,i}$ fixed to be $1$ for each $i\in S$).
%For any $\{0,1\}$-vectors $\xx$ and $\zz$ (with $x_{S,i}=1$ for all $i\in S$),
%  we use $\sol(\xx,\zz)$ to denote the \emph{feasible} solution uniquely determined by $\xx$ and $\zz$ 
%  using previous lemmas,
%  and use $\rev(\xx,\zz)$ to denote its expected revenue.
%The following lemma follows directly from Lemma \ref{lem:2}.
%  
%\begin{lemma}\label{hehe}
%Given any $\xx$ and $\zz$, we have in $\sol(\xx,\zz)$ that
%  $u_S\le |S|h$ and $u_\emptyset'\le \alpha$.
%\end{lemma}  
  
So far we have shown in an optimal solution to the relaxed IP that
  every entry in $\xx'$ and $\zz'$ is $1$ and
  the only unsettled variable in $\uu'$ is $u_\emptyset'$, with 
  $u_S'=u_\emptyset'+\sum_{i\in S} h_i$ for all $S\ne\emptyset$.
(For a sanity check, note that these conditions hold in  Solution 1 and 2.)
Next
  we prove the first case of Lemma \ref{lem:main}.
 %we show that if $z_\emptyset=1$ in an optimal solution,
  %then its revenue is at most that of Solution 1. 
  %(Since the re is a relaxation of the 
  %canonical LP, Solution 1 is a feasible solution to RLP.)
    
\begin{lemma}\label{lem:4}
In an optimal solution to the relaxed IP,
  if $z_\emptyset=1$, then it must be Solution 1.
\end{lemma}
\begin{proof}
Assuming $z_{\emptyset}=1$, we have $u_\emptyset'\ge u_\emptyset+\alpha$ from constraints (3). 
%and thus, $u_\emptyset'=\alpha$,
%  which is the largest possible value for $u_\emptyset'$ by Lemma \ref{hehe}.
Next we show that $x_{\emptyset,i}=1$ for all $i\in [n]$.
Assume for contradiction that $x_{\emptyset,i}=0$ for $i\in I$ and $I$ is nonempty.
Then we make the following changes to obtain a new solution: i) We change $x_{\emptyset,i}$ from
  $0$ to $1$ for every $i\in I$; and ii)
For each $S\ne \emptyset$, we change $u_S$ to be the maximum of the original
  $u_S$ and $u_\emptyset+\sum_{i\in S} h_i$ (by doing this $u_S$ can
  go up by at most $\sum_{i\in S\cap I}h_i$ due to constraints (5) for $T=\emptyset$).
 
We first verify that the new solution is also feasible and then show that
  its expected revenue is strictly higher than that of the old solution.
For its feasibility, constraints in (2) and (4) are trivial.
For (3), the constraint for $S\subseteq [n]$ is trivial if $u_S$ stays the same.
Otherwise, $u_S=u_\emptyset+\sum_{i\in S} h_i$ and
\begin{eqnarray*}
&\text{RHS of (3)}=u_S-\sum_{i\in S}h_i+\alpha\cdot z_S=u_\emptyset+\alpha\cdot z_S\le u_\emptyset+\alpha\le u_\emptyset'.&
\end{eqnarray*}
The constraint (5) for $T\subset S$ is trivial if $u_T$ remains the same.
Otherwise, we have $u_T=u_\emptyset+\sum_{i\in T} h_i$ and $u_S\ge u_\emptyset+ \sum_{i\in S} h_i$. Thus,
$u_S-u_T \ge \sum_{i\in S\setminus T} h_i\cdot x_{T,i}$.
This finishes the proof of feasibility.

Next the net gain of expected revenue in the new solution is at least
$$
|I|\cdot \Pr[\emptyset,\sigma]-\sum_{S\subseteq [n]} \Pr[S,\sigma]\cdot \sum_{i\in S\cap I}h_i
=(1-\tau+\eps)\cdot \left(|I|\cdot (1-p)^n-\sum_{i\in I} ph_i\right)>0,
$$
following the same argument used  in the previous lemma.
We conclude that $x_{\emptyset,i}=1$ for all $i\in [n]$.
%Let $j^* = \min\{ j | y_j =1 \}$. Set $y_j=1$ for all $i \geq j^*$, if they are not all already 1.
%By Lemma \ref{lem:2} the same values for the $u$ variables (and hence also $u'_0$ and the other $u'$ variables)
%still satisfy all the constraints,
%while the objective function increases in value.

Finally, given that $x_{\emptyset,i}=1$ for all $i\in [n]$ and $u_\emptyset\ge 0$, we have
  $u_S\ge \sum_{i\in S}h_i$, $u_\emptyset'\ge \alpha$ and $$u_S'= u_\emptyset'+\sum_{i\in S} h_i 
  \ge \alpha+\sum_{i\in S} h_i .$$
Comparing it with Solution 1, a) the utility of every valuation of this solution
  is at least as large as that of Solution 1; and b) every allocation variable $x_{S,i}$,
  $\smash{x_{S,i}'}$, $z_S$ and $z_S'$ in Solution $1$ is $1$ (the maximum possible).
Thus, its revenue is no more than that of Solution 1 and in order to be optimal,
  it must be exactly the same as Solution 1.
This finishes the proof of the lemma.
\end{proof}

In the rest of this subsection, we consider the more challenging case when $z_\emptyset=0$, and prove
  the second part of Lemma \ref{lem:main} using a sequence of lemmas.
  % to show that Solution 2 has the optimal expected
%  revenue among such solutions.
First we show that $u_\emptyset'=0$ when $z_\emptyset=0$.
  
\begin{lemma}\label{lem:5}
In any optimal solution to the relaxed IP, if $z_\emptyset=0$ then we have $u_\emptyset'=0$. %$ then $z_i=1$ for all $i \geq j$.
\end{lemma}
%\ignore{\begin{proof}
%Suppose $z_j=1$ for some $j$ and $z_i=0$ for some $i \geq j$.
%Set $z_i=1$, leaving all the other variables the same. 
%The only constraint we need to check is 
%$u'_0 \geq u_i - ihx_i + \alpha z_i = u_i -ih +\alpha z_i$ (since $x_i=1$ for all $i$).
%We know that $u'_0 \geq u_j - jhx_j + \alpha z_j = u_j -jh +\alpha$.
%Furthermore,  $u_i \leq u_j + (i-j)h$ by Lemma \ref{lem:2} (with equality iff $j^* \leq j$).
%Therefore, 
%$u_i -ih +\alpha z_i \leq u_j + (i-j)h  -ih +\alpha =u_j -jh +\alpha \leq u'_0$.
%Hence the new solution is feasible, and the objective value increases because the
%coefficient of $z_i$ is positive.
%\end{proof}}
\begin{proof}
Assume for contradiction that $u_\emptyset'>0$.
We show first that $u_\emptyset'>0$ implies that $u_\emptyset'\ge \delta$.

For this purpose we fix  
  all the allocation variables to their $\{0,1\}$ values in the optimal solution~and replacing each $u_S'$, $S\ne \emptyset$, using $u_S'=u_\emptyset'+\sum_{i\in S} h_i$ 
  (by Lemma \ref{lem:4})
  to obtain a linear program over the rest of utility variables
  $u_S$, $S\subseteq [n]$, and $u_\emptyset'$.
The linear program maximizes the objective function 
  (hence these utility variables must have the minimum possible values), subject to the constraints 
%   must be set to maximize the objective function,  
%subject to the constraints 
(1), (3), (4), (5) of the relaxed IP
(note that constraints (2) are already satisfied).   
All the constraints (3), (4), (5) of the relaxed IP have the form $v \geq v' + b$ 
where $v, v'$ are two variables from
the set $V= \{ u_S:S\subseteq [n]\} \cup \{u_\emptyset'\}$ and
$b$ has the form $c_1 h +  c_2 \delta$, where $c_1, c_2$ are integers
and $|c_2| \leq O(n 2^n)$ (note that all $h_i$ have this form
with $c_2 \in [0,2^n]$, $\alpha$ also has this form with $c_2 \in [0,(n/2)2^n]$,
and the allocation variables have value 0 or 1).
Recall that $h=2^{2n}$ and $\delta = 1/2^{3n}$.
The problem of finding the optimal utilities becomes a shortest-path problem over $2^n+2$ vertices: 
  the set $V$ plus an extra vertex $0$ as the source vertex.
The distance from $0$ to each vertex in $V$ is set to $0$ and 
  the distance from $v'$ to $v$ is set to $-b$ if $v\ge v'+b$ is a constraint in
  the relaxed IP.
The optimal value of $u_\emptyset'$ is then 
  the length of the shortest path from vertex $0$ to vertex $u_\emptyset'$ with the sign changed.
Hence $u_{\emptyset}'$ is a sum of at most $2^{n}+1$ edge weights ($-b$), with each of them 
  being an integer multiple $c_1$ of $h$ plus an integer multiple $c_2$  of $\delta$ (with $|c_2| \leq O(n2^n)$).
As a result, $u_\emptyset'$ is of the form $d_1h+d_2\delta$ where $d_1$ and $d_2$ are both integers
  and $\smash{|d_2|=O( n2^{2n})}$.
For $u_\emptyset'$ to be positive, we have either $d_1>0$, in which case $\smash{u_\emptyset'=\Omega(h)}$
  (as $\smash{\delta=1/2^{3n}}$), or $d_1=0$ and $d_2>0$, in which case $u_\emptyset'\ge \delta$.

Now we modify the given solution to obtain a new solution.
We then show that the new solution is also feasible and its expected revenue is strictly higher.
Let $\calL$ denote the set of all subsets $S\subset [n]$ with
  $\sum_{i\in S} h_i<\alpha$, and $\calU$ denote the set of all $S\subseteq [n]$ with $\sum_{i\in S} h_i\ge \alpha$.
We modify the solution as follows: 
\begin{enumerate}
\item[i)] We set $x_{S,i}=0$ and $z_S=0$ for all $S\in \calL$ and $i\notin S$;
\item[ii)] We change $u_\emptyset'$ to $0$ and $u_S'$ to $\sum_{i\in S}h_i$; 
\item[iii)] For each $S\in \calL$ we change $u_S$ to $0$;
 For each $S\in \calU$ we change $u_S$ to $\sum_{i\in S} h_i-\alpha$ (which\\ is nonnegative by
  the definition of $\calU$ and can only go down compared to the old solution\\
  by constraints (4) using $\smash{x_{\emptyset,i}'=1}$ from Lemma \ref{lem:3} and 
  $\smash{z_\emptyset'=1}$ from Lemma \ref{lem:1}).
\end{enumerate}
Every other entry remains the same as in the original solution.

We first show that the new solution is feasible.
Constraints in (2) are trivial.
For (3), if $S\in \calL$, we have $u_S=0$ and $z_S=0$ and thus, the RHS is 
  $-\sum_{i\in S} h_i\le 0=u_\emptyset'$;
  if $S\in \calU$, we have $u_S=\sum_{i\in S}h_i-\alpha$ and thus, the RHS is
  $-\alpha+\alpha\cdot z_S\le 0=u_\emptyset'$.
For constraints in (4), if $S\in \calL$, we have $u_S=0$ and the RHS is 
  $\sum_{i\in S}h_i-\alpha<0=u_S$ (by the definition of $\calL$);
  if $S\in \calU$  the RHS is 
  $\sum_{i\in S}h_i-\alpha=u_S$.~Finally for (5), the constraint is trivial if both $S$ of $T$ are in $\calU$, or both of them are in $\calL$.
The only case left is that $S\in \calU$ and $T\in \calL$ (the other case of $S\in\calL$ and $T\in \calU$
  cannot happen as $T\subset S$), in which case the constraint is trivial since $u_T=0$
  and $x_{T,i}=0$ for all $i\notin T$.
This finishes the  feasibility proof.  
  
Next we bound the net gain of expected revenue in the new solution. 
For this purpose we discuss two cases: a) There exists an $S$ with $|S|<n/2$ (and thus, $S\in\calL$)
  such that $z_S=1$ in the original solution; and b) Every $S$ with $|S|<n/2$ has $z_S=0$ in the 
  original solution.
  
We start with the first case.
The existence of such an $S$ implies that $u_\emptyset'\ge \Omega(h)$ in the original solution by constraint (3).
Note that the utility $u_S$ for every $S$ does not increase and 
  the utility $u_S'$ for every $S$ goes down by at least $\Omega(h)$.
Therefore the expected revenue goes up by at least
\begin{align*}
&\Omega(h)\cdot (\tau-\eps)-\Pr[\emptyset,\sigma]\cdot n-\sum_{\emptyset\ne S\in \calL}(n+\sigma)\cdot \Pr[S,\sigma]\\
&\hspace{1cm}\ge \Omega(h)-(1-\tau+\eps)\cdot n-O(\sigma)\cdot (1-\tau+\eps)\cdot (1-(1-p)^n)>0.
\end{align*}
For the last term, we plug in $1-\tau+\eps=O(\alpha/(\sigma+\alpha))=O(\alpha/\sigma)$ and $(1-p)^n\ge 1-pn$ to have 
$$
O(\sigma)\cdot (1-\tau+\eps)\cdot (1-(1-p)^n)=O\left(\sigma\cdot \frac{\alpha}{\sigma }\cdot pn\right)
=O(n^2)=o(h).
$$
%Since $u_{\emptyset}'$ is always a multiple of $h$ (recall that $n$ is even)
%  we have $u_\emptyset'\ge h$.
%We then modify $\xx$ and $\zz$ as follows.
%We set $x_{S,i}$ to $0$ for all $S\subseteq [n]$ and $i\notin S$
%  and set $z_{S}$ to $0$ for all $S\ne \emptyset$ (note that $z_\emptyset$ is already $0$ in the
%  original solution).
%Let $\xx^*$ and $\zz^*$ denote the new vectors 
%  (where $\zz^*$ is all zero and $x_{S,i}=1$ iff $i\in S$).
%  
%Note that $u_S=0$ for all $S\subseteq [n]$ in $\sol(\xx^*,\zz^*)$ so can only go down compared to the original %solution $\sol(\xx,\zz)$.
%More importantly, $u_\emptyset'$ in $\sol(\xx^*,\zz^*)$ must go down by at least $h$ since it is $0$
%  in the new solution and at least $h$ in the old solution.
%As a result, $\rev(\xx^*,\zz^*)-\rev(\xx,\zz)$ is at least
%$$
%h\cdot (\tau-\eps)-(1-\tau+\eps)\cdot n-(1-\tau+\eps) \cdot(1-\Pr[\emptyset])\cdot\sigma\ge 
%\Omega(h)-o(1)-(1-\tau+\eps)\cdot O(pn\sigma),
%$$
%%Note that $\eps/\tau$ is exponentially small in $n$.
%where we used $\tau\approx 1$, $1-\tau$ is exponentially small in $n$, $\eps=o(1/\sigma)$ with
%  $\sigma$ being exponentially large in $n$, and $p\approx 1/n^3$
%  (and thus, we have $1-\Pr[\emptyset]=O(pn)$). 
%We also have $1-\tau+\eps=O(\alpha/\sigma)$ and thus, 
%$$
%\rev(\xx^*,\zz^*)-\rev(\xx,\zz)\ge \Omega(h )-o(1)-O(pn\alpha)>0,
%$$
%using $\alpha=nh/2$ and $p=O(1/n^3)$.
%This finishes the proof of the lemma.

For the second case, we use $u_\emptyset'\ge \delta$.
Let $\calL'$ denote the set of $S$ with $|S|<n/2$.
Since $z_S=0$ for all $S\in \calD'$ in the original solution,
  the expected revenue goes up by at least
\begin{align*}
&\delta\cdot (\tau-\eps)-\sum_{S\in \calL'} n\cdot \Pr[S,\sigma]-\sum_{S\in \calL\setminus\calL'} 
  (n+\sigma)\cdot \Pr[S,\sigma]\\
&\hspace{1cm}\ge \Omega(\delta)-(1-\tau+\eps)\cdot n
  -O(\sigma)\cdot (1-\tau+\eps)\cdot \Pr[n/2].
\end{align*}
For the last term, we have $ {O(\sigma)\cdot (1-\tau+\eps)=O(\alpha)=O(n2^{2n})}$ and
$\Pr[n/2]<2^n\cdot p^{n/2}$. As a result, their product is much smaller than $ \delta=1/2^{3n}$
  as $ p=O(1/2^{2n})$.
This finishes the proof.
\end{proof}

%%In the rest of the proof we  $\calL$ denote the set of all subsets $S\subset [n]$ with
%  $\sum_{i\in S} h_i<\alpha$, and let $\calU$ denote the set of all $S\subseteq [n]$ with $\sum_{i\in S} h_i\ge %\alpha$.

Let $\calL$ and $\calU$ be the sets defined in the proof above.
We have the following simple corollary:

\begin{corollary}
In any optimal solution to the relaxed IP, if $z_\emptyset=0$ then $z_S=0$ for all $S\in \calL$.
\end{corollary}
\begin{proof}
By constraints (3), for $S\in \calL$:
$
0\ge u_S-\sum_{i\in S}h_i+\alpha\cdot z_S.
$
Thus, $\alpha\cdot z_S\le \sum_{i\in S}h_i<\alpha$. 
\end{proof}

Next we show that if $z_\emptyset=0$ in an optimal solution,
  then $x_{S,i}=0$
  for all $S\in \calL$ and $i\notin S$.
%This then can be used to show that $x_{S,i}=z_S=1$ for all $i\notin S$ and $|S|\ge n/2$,
%  which lead to Solution 2.
%For this purpose we need to set $\sigma$ in the rest of the proof to be
%$$
%\color{red}\sigma=\frac{3n^2}{\Pr[n/2]},
%$$
%which is exponentially large in $n$ as promised at the beginning of the reduction.

\def\calG{\mathcal{G}} \def\calD{\mathcal{D}}

\begin{lemma}
In an optimal solution, if $z_\emptyset=0$ then $x_{S,i}=0$
  for all $S\in \calL$ and $i\notin S$.
\end{lemma}
\begin{proof}
%We first prove 1) and then use it to prove 2).
%The part of $z_S=0$ is trivial since otherwise we have $z_\emptyset'>0$, contradicting
%  with the assumption that $z_\emptyset'=0$.
%We prove $x_{S,i}=0$ for all $i\notin S$ and $|S|<n/2$ in the rest of the proof.
%Let $\calD'$ denote the set of $S\subset[n]$ with $|S|<n/2$, and let $\calD^*$
%  denote the set of $S\subset[n]$ with $|S|=n/2$ and $\sum_{i\in S}h_i<\alpha^*$.
%Then we have $\calD=\calD'\cup\calD^*$.
Assume for contradiction that we do not have $x_{T,i}=0$ for all sets $T\in \calL$ and $i\notin T$. 
(For reasons that will become clear later we need special treatments for not only 
  those $T$ at level $n/2$ but those at level $(n/2)-1$ as well.)
We use $\calG$ to denote the set of pairs $(T,k)$ such that $|T|<(n/2)-1$, $k\notin T$ and $x_{T,k}=1$;
we use $\calG^*$ to denote the set of pairs $(T,k)$ such that $|T|=(n/2)-1$, $k\notin T$,
  $T\cup\{k\}\in \calU$  and $x_{T,k}=1$;
we use $\smash{\calG^\dag_1}$ to denote the set of pairs $(T,k)$ such that $|T|=(n/2)-1$, $k\notin T$,
  $T\cup\{k\}\in \calL$ (so it follows from the first condition that $B$ and $W$ satisfy that 
  $n\notin T\cup\{k\}$) and $x_{T,k}=1$;
  we use $\smash{\calG^\dag_2}$ to denote the set of pairs $(T,k)$ such that $|T|=n/2$, $T\in \calL$, $k\notin T$ and $x_{T,k}=1$.
Then $\calG\cup \calG^* \cup\smash{\calG^\dag_1\cup \calG^\dag_2}$ is nonempty.
We next use $\calG, \calG^*,\smash{\calG^\dag_1}$ and $\smash{\calG^\dag_2}$ to define two sets $\calE$ and $\smash{\calE^\dag}$, 
  where $\calE$ consists of subsets of size $n/2$ and $\smash{\calE^\dag}$ consists of subsets of size $n/2+1$
  so they are disjoint. 
\begin{flushleft}\begin{enumerate}
\item $S\subset [n]$ is in $\calE$ if $|S|=n/2$, $ S\in \calU$ and $T\cup\{k\}\subseteq S$
  for some pair $(T,k)\in \calG\cup \calG^*$.
\item $S\subset [n]$ is in $\smash{\calE^\dag}$ if $|S|=(n/2)+1$ and satisfies either $S=T\cup\{k\}$ for some 
  pair $(T,k)\in \smash{\calG^\dag_2}$ or $S=T\cup\{k,r\}$ for some pair $(T,k)\in \smash{\calG^\dag_1}$ with
  $r$ being the smallest index not in $T\cup\{k\}$.
\end{enumerate}\end{flushleft}

We need the following simple claim about $z_S$, $S\in \calE\cup\calE^\dag$, in the original solution.

\begin{claim}\label{hahaclaim}
We have $z_S=0$ for every $S\in \calE\cup \calE^\dag$ in the original solution.
\end{claim}
\begin{proof}
For each $S\in \calE$ (with $T\cup \{k\}\subseteq S$ for some $(T,k)\in 
  \calG\cup\calG^*$)  we have from constraint (5): $$u_S\ge u_T+\sum_{i\in S\setminus T} h_i\cdot x_{T,i}
\ge h_k$$ and thus, by constraint (3):$$0=u_\emptyset'\ge u_S-\sum_{i\in S}h_i+\alpha\cdot z_S
\ge-((n/2)-1)h-\sum_{i\in S\setminus\{k\}} a_i+\alpha \cdot z_S, $$
which implies that $z_S=0$  since $\alpha \ge (n/2)h$ and $h\gg a_i$.

For each $S\in \smash{\calE^\dag}$, if $S=T\cup\{k\}$ for some $(T,k)\in \smash{\calG^\dag_2}$, we have $u_S\ge h_k$ and
  by constraint (3): 
$$
0=u_\emptyset'\ge u_{S}-\sum_{i\in S}h_i+\alpha \cdot z_S
\ge -\sum_{i\in T}h_i+\alpha \cdot z_S,
$$
which implies that $z_S=0$ since $T\in \calL$ implies that $\sum_{i\in T}h_i<\alpha $.

Finally if $S\in \smash{\calE^\dag}$ satisfies $S=T\cup\{k,r\}$ for some $(T,k)\in \smash{\calG^\dag_1}$,
  we have $u_S\ge h_k$ and thus,
$$
0=u_\emptyset'\ge u_S-\sum_{i\in S}h_i+\alpha \cdot z_S\ge-\sum_{i\in T\cup\{r\}}h_i+\alpha \cdot z_S.
$$
Note that $n\notin T$ (by the first condition of $B$ and $W$ and the fact that $T\cup\{k\}\in \calL$), $r\ne n$, and $T\cup\{r\}$ contains either $1$ or $2$ (by our choice of $r$, as the smallest index not in $T\cup\{k\}$).
Using the second condition about $B$ and $W$, we have
$\sum_{i\in T\cup\{r\}}h_i<\alpha $ and thus, $z_S=1$.
This finishes the proof.
\end{proof}

Claim \ref{hahaclaim} inspires us to derive a new solution by making the following changes in the old one:
%that this is not the case. 
%Then we compare $\rev(\xx,\zz)$ with $\rev(\xx^*,\zz^*)$ where
%  $(\xx^*,\zz^*)$ is defined as follows.
%Let $\calG$ denote the set of pairs $(S,i)$ with $|S|<n/2$, $i\notin S$, and $x_{S,i}=1$.
%Let $\calT$ denote the set of $T\subseteq [n]$ such that $|T|=n/2$ and for some 
%  $(S,i)\in \calG$, $S\subset T$ and $i\in T$.
\begin{flushleft}\begin{enumerate}
\item[i)] For all $\smash{(T,k)\in \calG\cup \calG^*\cup \calG^\dag_1 \cup\calG^\dag_2}$, change $x_{T,k}$ from $1$ to $0$ 
  (so that $x_{T,k}=0$ in the new solution for all $T\in \calD$ and $k\notin T$);
\item[ii)] For all $S\in \calE \cup\calE^\dag$, 
  change $z_S$ from $0$ to $1$;
\item[iii)] For all $S\in \calD$, change $u_S$ to $0$; For all $S\in \calU$, change $u_S$ to $\sum_{i\in S}h_i-\alpha $ (note that the new $u_S$ is nonnegative as 
  $S\in \calU$ and can only go down from the original by constraint (4)).
\end{enumerate}\end{flushleft}
All other entries remain the same in the new solution.

We first verify that the new solution is feasible and then show that it is strictly better than the 
  old solution.
For the feasibility, constraints (2) are trivial.
For (3), the constraint is trivial if $S\in \calD$ (since $u_S=z_S=0$);
  if $S\in \calU$, we have $u_S=\sum_{i\in S}h_i-\alpha $ and thus, the RHS is at most $0$.
For (4), the constraint is trivial if $S\in \calD$ (since the RHS is negative);
  if $S\in \calU$, the LHS and RHS are the same.
For (5), the constraint is trivial if $T\in \calD$ (since the RHS is $0$);
  otherwise we have both $S$ and $T$ are in $\calU$ and the constraint follows from
  $u_S=\sum_{i\in S}h_i-\alpha $ and $u_T=\sum_{i\in T}h_i-\alpha $ in the new solution.
This finishes the proof of feasibility of the new solution.

Note that each utility variable in the new solution can only go down
  from that in the old solution.
As a result, the net gain of expected revenue in the new solution is at least
$$
 \sum_{S\in\calE\cup \calE^\dag} \sigma\cdot \Pr[S,\sigma]-\sum_{(T,k)\in \calG\cup\calG^*\cup \calG^\dag_1\cup
  \calG^\dag_2} \Pr[T,\sigma]
.%=(1-\tau+\eps)\left(\sigma \sum_{T\in \calT}\Pr[T]- \sum_{(S,i)\in \calG} \Pr[S]\right).
$$
Ignoring the common factor of $1-\tau+\eps$ and rearranging the terms, we obtain
$$
\left(\sigma\sum_{S\in \calE}\Pr[S]-\sum_{(T,k)\in \calG\cup\calG^*}\Pr[T]\right)+
\left(\sigma\sum_{S\in \calE^\dag}\Pr[S]-\sum_{(T,k)\in \calG^\dag_1\cup\calG^\dag_2}\Pr[T]\right).
$$
Below we show that the first term is positive if $\calG\cup\calG^*$ is nonempty and the second term
  is positive if $\smash{\calG^\dag_1\cup\calG^\dag_2}$ is nonempty. %\ne \emptyset$.
This shows that the expected revenue goes up strictly in the new solution.

We start with the second term which is easier. By the definition of $\smash{\calE^\dag}$ 
  from $\smash{\calG^\dag_1}$ and $\smash{\calG^\dag_2}$,
  we have $$|\calG^\dag_2|\le ((n/2)+1)\cdot |\calE^\dag|\quad \text{and}\quad |\calG^\dag_1|\le ((n/2)+1)\cdot (n/2)\cdot
    |\calE^\dag|.$$
Thus, the second term is at least
\begin{align*}
%&\sigma\sum_{S\in \calS^*}\Pr[S]-\sum_{(T,k)\in \calG^*}\Pr[T]\\
\sigma\cdot |\calE^\dag|\cdot p^{(n/2)+1} (1-p)^{(n/2)-1}
-n\cdot |\calE^\dag|\cdot p^{n/2} (1-p)^{n/2}-n^2\cdot |\calE^\dag|\cdot p^{(n/2)-1} (1-p)^{(n/2)+1}>0 
\end{align*}
when $\calE^\dag\ne \emptyset$ since $p^2\sigma\gg pn+n^2$.
This finishes the proof, as $\calE^\dag\ne \emptyset$ whenever $\smash{\calG^\dag_1\cup\calG^\dag_2\ne\emptyset}$.

For the first term, note that $|\calG^*|\le (n/2)\cdot |\calE|$ and thus,
$$
\sum_{(T,k)\in \calG^*}\Pr[T]\le (n/2)\cdot |\calE|\cdot p^{(n/2)-1} (1-p)^{(n/2)+1}.
$$
To understand the sum over $(T,k)\in \calG$, we
 decompose $\calG$ into $\calG_0,\ldots,\calG_{(n/2)-2}$ where $\calG_\ell$   
  contains all the pairs $(T,k)\in \calG$ with $|T|=\ell$.
Fixing an $\ell$,
  we examine the following bipartite graph between $\calG_\ell$ and $\calE$: $(T,k)\in \calG_\ell$ and 
  $S\in \calE$ (and thus, $S\in \calU$) are connected if $ T\cup\{k\}\subseteq S$.
We claim that every $(T,k)\in \calG_\ell$ has at least one edge (it has a lot more 
   but one is enough for our purpose).
    
To see this, we consider two cases: $n\in T\cup\{k\}$ or $n\notin T\cup\{k\}$.  
For the first case, from the first condition on $B$ and $W$, every set $S$ of size $n/2$ that 
  contains $T\cup\{k\}$ is in $\calE$ and is connected with $(T,k)$.
For the second case, every set $S$ of size $n/2$ that contains $T\cup\{k,n\}$ is in $\calE$ and is connected 
  with $(T,k)$.
(Here we used the fact that $\ell\le n/2-2$ so $T\cup\{k,n\}$ has size no more than $n/2$, which is why we treated sets of size
  $n/2-1$ differently throughout the proof.)
%Then we make the following observation on the degree of each $(T,k)\in \calG_\ell$.
%If $n\in T\cup \{k\}$, then every size-$(n/2)$ $S$ that contains $T\cup\{k\}$ is in $\calE$
%  (using the first condition on $B$ and $W$).
%Thus, the degree of $(T,k)$ is 
%$$
%{n-\ell-1\choose n/2 -\ell-1}.
%$$
%If $n\notin T\cup\{k\}$, every size-$(n/2)$ $S$ that contains $T\cup\{k,n\}$ is in $\calE$.\footnote{.} So the degree of $(T,k)$ is at least
%$$
%{n-\ell-2\choose n/2-\ell-2}.
%$$

As a result, the number of edges between $\calG_\ell$ and $\calE$ is at least
$
|\calG_\ell|
$
but is at most
$$
|\calE|\cdot \binom{n/2}{\ell}\cdot (n/2-\ell).
$$
Therefore, we have 
$$
|\calE|\ge \frac{|\calG_\ell|}{\binom{n/2}{\ell}\cdot (n/2-\ell)}
> |\calG_\ell|\cdot \frac{1}{n2^n}.
$$
This implies that 
\begin{align*}
\frac{\sum_{S\in \calE} \Pr[S]}{\sum_{(T,k)\in\calG_\ell}\Pr[T]}
&=\frac{|\calE|\cdot p^{n/2}\cdot (1-p)^{n/2}}{|\calG_\ell|\cdot p^\ell\cdot (1-p)^{n-\ell}} 
 > \frac{p^{n/2}}{2p^\ell}\cdot \frac{1}{n2^n}=\frac{p^{n/2}}{2n2^n p^\ell}.
%\\ &=\frac{p^{n/2}\cdot (1-p)^{n/2}}{p^\ell}\cdot \frac{(n-\ell-2)\cdots (n/2+1)\cdot \ell! }
%{(n/2-\ell-2)!\cdot (n/2)\cdots (n/2-\ell+1)\cdot (n/2-\ell)} \\[0.5ex]
%& =\frac{p^{n/2}\cdot(1-p)^{n/2}}{p^\ell}\cdot {n \choose n/2}\cdot \frac{(n/2-\ell-1)\cdot \ell!}{n\cdot %(n-1)\cdots (n-\ell-1)}\%\
%&\ge \frac{{\Pr[n/2]}}{n^2\cdot (np)^\ell}. 
\end{align*}
Therefore, we have
$$
 \sum_{(T,k)\in\calG}\Pr[T]=\sum_{\ell=0}^{n/2-2} \sum_{(T,k)\in \calG_\ell} \Pr[T]< \left(\frac{2n2^n}{p^{n/2}}\cdot \sum_{S\in \calE} \Pr[S]\right)\cdot \sum_{\ell=0}^{n/2-2}p^\ell<
 \frac{4n2^{n}}{p^{n/2}} \cdot\sum_{S\in \calE} \Pr[S] . 
$$
Plugging in our choice of $\sigma=1/p^n$, we have
$$
\sigma\sum_{S\in \calE}\Pr[S]-\sum_{(T,k)\in \calG\cup\calG^*}\Pr[T]
\ge \left(\sigma-\frac{4n2^n}{p^{n/2}}-O(n/p)\right)\cdot 
\sum_{S\in \calE}\Pr[S]>0
$$
when $\calE\ne\emptyset$. 
Our argument above also shows that $\calE\ne\emptyset$ whenever $\calG\cup \calG^*\ne\emptyset$.
% This finishes the proof.
\end{proof}

Finally we %show that if $z_\emptyset=0$ in an optimal solution then its no better than Solution 2.
  prove the second part of Lemma \ref{lem:main}.
Lemma \ref{lem:main} follows from Lemma \ref{lem:4} and \ref{hehe2}.
\begin{lemma}\label{hehe2}
In an optimal solution to the relaxed IP,
  if $z_\emptyset=0$, then it must be Solution 2.
\end{lemma}
\begin{proof}
We show that, when $z_\emptyset=0$, every utility variable in the optimal solution is at least as large as that in Solution 2,
  and every allocation variable is at most as large as that in Solution 2. So
  for it to be optimal, it must be exactly the same as Solution 2.
%The lemma then follows. 

For utilities we first note that $u_S'$ is the same in both solutions for all $S\subseteq [n]$ since
  in both we have $u_\emptyset'=0$ and $\smash{u_S'=\sum_{i\in S}h_i}$.
Next for each $S\in \calD$, we have $u_S=0$ in Solution 2.
Finally, for each $S\in \calU$ we have from constraint (4) that $u_S\ge u_\emptyset'+\sum_{i\in S}h_i-\alpha
  =\sum_{i\in S}h_i-\alpha$ in the optimal solution, which is at least as large as $u_S$ in Solution 2.
For allocation variables, we first have $\smash{x_{S,i}'=z_S'=1}$ in both solutions for all
  $S\subseteq [n]$ and $i\in [n]$.
Next we have $x_{S,i}=1$ in both solutions for all $S\subseteq[n]$ and $i\in S$.
For each $S\in \calD$, we have $x_{S,i}=z_S=0$ in both solutions for all $i\notin S$.
Finally for each $S\in \calU$, we have $x_{S,i}=z_S=1$ in Solution 2 for all $i\notin S$.
This finishes the proof of the lemma.%It then follows from our characterization of $u_S$ that
%  $u_S=0$ for all $|S|<n/2$ and 
%  $u_S\le (|S|-(n/2))h$ for all $|S|\ge n/2$.
%As a result, if $z_S=0$ for some $|S|\ge n/2$, we can change it to $z_S=1$
%  and still have $u_\emptyset'=0$.
%This implies that $z_S=1$ for all $S$ with $|S|\ge n/2$.
%Next assume for contradiction that $x_{T,i}=0$ for some $i\notin T$ and $|T|\ge n/2$.
%Then we use $\xx^*$ to denote the vector obtained from $\xx$ by changing $x_{T,i}$ to $1$.
%Comparing  $\sol(\xx^*,\zz^*)$ with $\sol(\xx,\zz)$, we have
%  1) $u_\emptyset'$ remains to be $0$ (since $u_S\le (|S|-(n/2))h$ holds for $|S|\ge n/2$
%  in the new solution as well); 2) the utility $u_S$ remains the same unless $T\subset S$
%  and $i\in S$, in which case the utility can go up by at most $h$.
%As a result,
%$$
%\rev(\xx^*,\zz^*)-\rev(\xx,\zz)\ge \Pr[T,\sigma]-h\sum_{S:T\subset S,i\in S}\Pr[S,\sigma]
%\ge (1-\tau+\eps)p^{|T|}\left((1-p)^{n-|T|}-hp\right)>0,
%$$
%using a similar argument as in (\ref{haha}).
%This finishes the proof of the lemma.
\end{proof}

\def\calR{\mathcal{R}} \def\calL{\mathcal{L}}

%Lemma \ref{lem:main} now follows from Lemma \ref{lem:4} and Lemma \ref{hehe2}.

\subsection{Finishing the reduction from \texorpdfstring{COMP$^*$}{COMP*}}\label{finish}

%{\color{blue}{blablabla}}

Finally, we show that with an appropriate choice of $\eps$ 
  (with $|\eps|=o(1/\sigma)$ as promised) that can be computed in polynomial time,
  we have 1) Solution 2 is strictly better than Solution 1 if $(B,W,t)$ is a yes-instance of
  COMP$^*$; and 2) Solution 1 is strictly better than Solution 2 if it is a no-instance.

%The two possibilities correspond to the two special solutions:
%Solution 1 corresponds to $z_\emptyset=1$ and solution 2 to $z_\emptyset=0$.
%It remains to show that we can set the parameter $\eps$ (with $|\eps|=o(1/\sigma)$ as promised earlier) 
%  so that the comparison of the two solutions shows us whether
%  the number of $(n/2)$-sets $S\subset [n]$ such that $\sum_{i\in S}h_i\ge \alpha$ is at least $t$
%  or at most $t-1$ (or equivalently, whether the number of $(n/2)$-sets $S\subset [n]$ such that
%  $\sum_{i\in S}b_i\ge w$ is at least $t$ or at most $t-1$).
The expected revenue of Solution 1 is $\rev_1 = n+\sigma$.
Let $t^*$ be the number of $(n/2)$-sets $S\subset[n]$ with $\sum_{i\in S}h_i\ge \alpha$
  (recall that $(B,W,t)$ is a yes-instance if $t^*\ge t$ and is a no-instance if $t^*\le t-1$)  
  and let $\calR$ be the set of all sets $S\subset[n]$ with $\sum_{i\in S} h_i<\alpha$.
Then $\rev_2$ of Solution 2 is
\begin{align*}
\rev_2&= (n+\sigma+\alpha)\left((\tau -\epsilon) + (1-\tau+\epsilon) \left( \Pr [i>n/2]
  +t^*\cdot p^{n/2}\cdot (1-p)^{n/2}\right)\right)\\
&\hspace{1cm}+(1-\tau+\epsilon) \sum_{S\in \calR} \Pr[S] \sum_{i\in S} (h_i+1).
\end{align*}
Below we use $a=b\pm c$ to denote $|a-b|\le c$.
%Here we use the notation $\Pr[i]$ to denote the probability that there are $i$ nonspecial items 
%at the high value, and the notation $\Pr [i\geq n/2]$ to denote the probability that the number $i$
%of nonspecial items at high value is at least $n/2$.
The sum $\sum_{S\in \calR}\Pr[S]\cdot \sum_{i\in S} (h_i+1)$ is equal to $$
\sum_{i\in [n]}(h_i+1)p-\sum_{S\notin \calR} \Pr[S] \sum_{i\in S} (h_i+1)
=\sum_{i\in [n]}(h_i+1)p\pm O(2^np^{n/2}nh).
$$
%since the second sum is negligible compared to the first one.
%=(h+1)np - (n/2)(h+1)\cdot \Pr[n/2] - \sum_{i>n/2} i(h+1)\cdot\Pr[i].$$
%We can also write $\Pr [i\geq n/2]$ as $\Pr[n/2] + \Pr[ i> n/2]$.
%Ignoring levels that are strictly above $n/2$ and ignoring  $\epsilon$
%(which will be set very small), 
Then $\rev_2$ is equal to $B-A\eps$, where
\begin{align*}
 A&= (n+\sigma+\alpha)\left(1-\Pr[i> n/2]-t^*\cdot p^{n/2}\cdot (1-p)^{n/2}\right)-\sum_{S\in \calR}\Pr[S]
\sum_{i\in S}(h_i+1)\\[-0.6ex]
& =\left((n+\sigma+\alpha)\big(1-\Pr[i>n/2]\big)-\sum_{i\in [n]}(h_i+1)p\right)\pm
  O(\sigma2^np^{n/2})\\
&=A'\pm O(\sigma2^np^{n/2}),
\end{align*}
where $A'=\Theta(\sigma)$ is a number that we can compute efficiently.
%and (letting $C= (1-\tau)\sum_{i>n/2}(n+\sigma+\alpha-i(h+1))\cdot \Pr[i]$ below)
On the other hand,  we note for $B$ that $\Pr[i>n/2]\le 2^np^{n/2+1}$ and $(n+\sigma+\alpha)\cdot (1-\tau)=O(\alpha)$. As a result, we have
\begin{align*}
B&=(n+\sigma+\alpha)\left(\tau+(1-\tau)\cdot t^*\cdot p^{n/2}\cdot (1-p)^{n/2} \right)
  +(1-\tau)\sum_{i\in [n]} (h_i+1)p\pm O(\alpha2^np^{n/2+1}).
%\tau(n+\sigma+\alpha) + (1-\tau)(h+1)np + (1-\tau) (\sigma+(n/2))\cdot \Pr [n/2] +C\\[0.3ex]
%&= \frac{\sigma}{\sigma +\alpha}(\sigma+\alpha) + \frac{\sigma}{\sigma+\alpha}\cdot n +\frac{\alpha}{\sigma+\alpha} \frac{(h+1)n}{h+n} + \frac{\alpha}{\sigma+\alpha}\cdot(\sigma+(n/2))\cdot \Pr[n/2]+C\\[1ex]
%&\hspace{2cm}+\frac{\alpha}{\sigma+\alpha}\sum_{i>n/2}(n+\sigma+\alpha-i(h+1))\Pr[i]\\
%&= \sigma +n - \frac{\alpha}{\sigma+\alpha}\cdot \frac{n(n-1)}{h+n} + \frac{3n\alpha}{\sigma+\alpha}
%  +\frac{\alpha(n/2)\cdot \Pr[n/2]}{\sigma+\alpha}+C.
%&\hspace{2cm}+\frac{\alpha}{\sigma+\alpha}\sum_{i>n/2}(n+\sigma+\alpha-i(h+1))\Pr[i].
\end{align*}
For convenience we write $B$ as $B=B'+C'\cdot t^*\pm O(\alpha2^np^{n/2+1})$, where 
$$
B'=(n+\sigma+\alpha )\tau+(1-\tau)\sum_{i\in [n]} (h_i+1)p\quad\text{and}\quad
C'=(n+\sigma+\alpha)\cdot(1-\tau)\cdot p^{n/2}\cdot(1-p)^{n/2}
$$
can also be computed efficiently.
Plugging in $\tau=\sigma/(\sigma+\alpha)$ and $h_i=h+a_i$, we have
\begin{align*}
B'&=(n+\sigma+\alpha)\tau+(1-\tau)n(h+1)p+(1-\tau)p\sum_{i\in [n]} a_i\\[0.3ex]
&=\sigma+\frac{n\sigma}{\sigma+\alpha}+\frac{\alpha}{\sigma+\alpha}\cdot n(h+1)\cdot \frac{1}{2(h+1)}+ (1-\tau)p\sum_{i\in [n]} a_i\\[1ex]
&=\sigma+n-\frac{\alpha n }{2(\sigma+\alpha)}+(1-\tau)p\sum_{i\in [n]} a_i.
\end{align*}

Finally we choose $\eps$ to be (recall $t$ is between $1$ and $2^n$; otherwise the problem is trivial)
$$
\eps=\frac{1}{A'}\cdot \left(C'\big(t-(1/2)\big)-\frac{\alpha n }{2(\sigma+\alpha)}+(1-\tau)p\sum_{i\in [n]} a_i\right),
$$
which can be computed efficiently and satisfies (using $|A'|=\Theta(\sigma)$)
  $$|\eps|\le (1/A')\cdot \big(O(\alpha p^{n/2}2^n)+O(\alpha n/\sigma)\big)=O(\alpha p^{n/2}2^n/\sigma)=o(1/\sigma).$$% as promised earlier.

%If $(B,W,t)$ is a yes-instance, we have $t^*\ge t$ and thus, $\rev_2-\rev_1$ is at least
%\begin{align*}
%&\left(-\frac{\alpha n }{2(\sigma+\alpha)}+(1-\tau)p\sum_{i\in [n]} a_i+C'\cdot t-
%  \big(A'\pm O(\sigma2^np^{n/2})\big)\cdot \eps\right)\pm O(\alpha2^np^{n/2+1})\\[1ex] 
%  &\hspace{1.5cm}=C'/2\pm  O(\alpha 2^np^{n/2+1} ) >0
%\end{align*}
Plugging in our choice of $\eps$, we have 
$ 
\rev_2-\rev_1=C'(t^*-t+1/2)\pm O(\alpha2^np^{n/2+1}).
$ 
%since $C'=\Omega(\alpha p^{n/2})$ and $p=O(1/2^{2n})$.
Note that $C'=\Omega(\alpha p^{n/2})\gg O(\alpha 2^n p^{n/2+1})$.
If $t^*\ge t$, Solution 2 is strictly better than Solution 1;
  if $t^*\le t-1$, Solution 1 is strictly better.  
This finishes the proof of Theorem \ref{thm:main}.
%$\rev_2-\rev_1\le %-\frac{\alpha n(1-(1/h))}{2(\sigma+\alpha)}\pm O(n2^n\delta)+C'\cdot (t-1)-
%  %(A'\pm O(\sigma2^np^{n/2}))\cdot \eps
% -C'/2\pm  O(\alpha  2^np^{n/2+1} ) >0.$
%This implies that 1) if the number of $(n/2)$-sets $S\subset [n]$ such that $\sum_{i\in S} b_i\ge w$
%  is at least $t$, the optimal expected revenue is strictly larger than $\sigma+n$; and 2)
%  if the number of such sets is at most $t-1$, the optimal expected revenue is equal to $\sigma+n$. 
%This finishes the reduction.

%For $C$, using $\Pr[i]=O(p\Pr[n/2])$ for all $i>n/2$, 
%$$
%C\le \frac{\alpha}{\sigma+\alpha}\cdot O(n\cdot \sigma\cdot p\Pr[n/2])=\frac{\alpha}{\sigma+\alpha}
%\cdot O(n^2p)=\frac{\alpha}{\sigma+\alpha}\cdot o(1).
%$$
%Using $h=n^3$ and that $\Pr[n/2]$ is exponentially small in $n$, we have that $$B=\sigma+n+\Theta(n\alpha/%(\sigma+\alpha)).$$
%As a result, to make $\rev_1=\rev_2$, we can just set
%$$
%\eps=\frac{B-(\sigma+n)}{A}=\Theta\left(\frac{n\alpha}{(\sigma+\alpha)\sigma}\right)=o(1/\sigma)
%$$
%as promised earlier.
%If we set $\sigma = \frac{n(n-1)}{(h+n) \Pr[n/2]} - \frac{n}{2}$, 
%then $rev2 \approx n+\sigma = rev1$; if we set $\sigma$ somewhat higher then $rev2 > rev1$
%and if we set it somewhere lower then $rev2 < rev1$ (taking into account
%also the terms we ignored).
%Clearly, we can pick a value for $\sigma$ that equates the two revenues;
%note that the value is exponential in $n$ because of the term $1/\Pr[n/2]$.

% !TEX root = main.tex
\def\calW{\mathcal{W}}
\def\supp{\textsc{support}}
\def\b{B}
\def\pr{\text{Pr}}

\section{Constant Number of Items}

In this section we prove Theorem \ref{thm:constant}. Let $\calF=\calF_1\times\cdots\times\calF_k$ be an instance of bundle-pricing for~some constant number of items $k$ and  assume without loss of generality that $|\supp(\calF_i)|=m$ for all \mbox{$i\in[k]$.} In this case, there are $m^k$ possible valuation vectors (a polynomial number), and $d=2^k$ possible distinct bundles (a constant number).
{The standard IP in this case has a polynomial number of variables and constraints. However, Integer Programming is NP-hard, so we will use a different method to solve the problem in polynomial time.} For the rest of this section, we assume two arbitrary orderings, one for the valuation vectors and one for the bundles, and we will use $v_i$ to denote the $i$th valuation vector, and $\b(j)$ to denote the $j$th bundle and $p_j$ to denote its price.

We will argue that we can generate in polynomial time a set of price vectors $p$ that includes an optimal one; we can then compute the expected revenue for each of these vectors and pick the best one.
To this end, we consider a partitioning of the $d$-dimensional
space of possible price vectors $p$ into cells, such that for all $p$ in the same cell, the buyer has the same behavior for every $v_i$, i.e., buys the same bundle, if any. Consider the following set $H$ of hyperplanes over $p$. 
\begin{flushleft}\begin{enumerate}
\item For each valuation $v_i$ and bundle $\b(j)$, the set $H$
includes $\sum_{\ell\in \b(j)} v_{i,\ell}-p_j =0$.
(If the price vector $p$ is below the hyperplane, the buyer will not consider bundle $\b(j)$ for valuation $v_i$.)\\
\item For each valuation $v_i$ and each pair of bundles $\b(j)$ and $\b(j')$,
the set $H$ includes the hyperplane 
$\smash{\sum_{\ell\in \b(j)} v_{i,\ell}-p_j = \sum_{\ell\in \b(j')} v_{i,\ell}-p_{j'}}$.
Note that for $v_i$, the buyer prefers $\b(j)$ to $\b(j')$ if $p$ is on one side of the hyperplane, she prefers $\b(j')$ to $\b(j)$ if $p$ is
on the other side, and if $p$ lies on the hyperplane itself then
it depends on the order between the prices $p_j$, $p_{j'}$.\\
\item For each pair of bundles $\b(j)$ and $\b(j')$, the set $H$ includes
the hyperplane $p_j = p_{j'}$.\\
%%(Prices must be nonnegative).
\end{enumerate}\end{flushleft}

These hyperplanes partition the space of prices into
cells, where a cell consists of all price vectors that
have the same relation to each of these hyperplanes, i.e., lie in
the same open half-space or on the hyperplane.
We can assume that for a valuation $v_i$ and price vector $p$, if there is a tie both in utility and in the price between some
bundles, then the buyer selects a bundle according to some fixed tie-breaking rule, for example she chooses among the tied bundles the one with the smallest index (the rule does not matter for the revenue). 
It follows from the definition of the set $H$ of hyperplanes, that for every cell $C$ and every valuation $v_i$ there is $k_i\in[d]$ such that the buyer selects the same bundle $\b(k_i)$ for every price vector $p$ in $C$, or buys no bundle
(if they all have negative utility).
Let $V_C(j)$ be the set of valuations for which the buyer selects bundle $\b(j)$ if the price vector $p \in C$,
and let $Q_{C}(j) = \pr[ V_C(j)]= \sum\{ \Pr[v_i] | v_i \in V_C(j) \}$ be the probability that the buyer selects bundle $\b(j)$.
The supremum revenue that the seller can extract for a price vector 
$p$ in the cell $C$, can be computed by solving the LP of maximizing
$\sum_j Q_{C}(j) \cdot p_j$, subject to $p$ belonging to the closure of
the cell $C$, i.e., $p$ satisfying all the weak inequalities corresponding to the bounding hyperplanes of $C$.
By LP theory, the maximum value of the LP is achieved at some vertex; even if the vertex does not belong
to $C$ but is in the closure, the corresponding price vector
achieves this expected revenue (by the maximum price tie breaking rule).
The maximum over all cells $C$ gives the supremum revenue that
can be achieved by any price vector. Thus, the supremum revenue is
achieved at some vertex,
i.e., at the intersection of some $d$ hyperplanes of the set $H$.

Therefore, we can compute an optimal solution 
by generating all vertices and picking the best one.
For every subset of $d=2^k$ hyperplanes of $H$, solve the corresponding linear system of equations to check if the hyperplanes
intersect at a unique point $p$, and if $p$ is nonnegative
(if a price is negative then $p$ cannot be optimal).
If so, compute the expected revenue of $p$ by examining each valuation $v_i$ and determining the bundle selected for $v_i$, if any.
Choose among these  price vectors $p$ the one that yields the maximum revenue.
The set $H$ has a polynomial number of hyperplanes, and the dimension
$d=2^k$ is constant, so we will consider a polynomial number of
subsets to generate the set of price vectors $p$. 
Since the number of valuations is polynomial, it takes polynomial time to compute the revenue of each vector $p$.
Hence the total time is polynomial. 
This finishes the proof of Theorem \ref{thm:constant}.\vspace{-0.1cm}
\def\rev{\textsc{Rev}}
\section{Conclusions}

In this work, we studied the optimal bundle-pricing problem {(or equivalently, the Revenue-Optimal Deterministic Mechanism Design} problem). We showed that the problem is intractable (\#P-hard) even when the (independent) item distributions have support size 2 and the optimal solution has a very simple form of discounted item-pricing: the seller prices the individual items and offers also the grand bundle at a (possibly) discounted price. Another consequence of the results is that there is no `nice' (easy-to-check) characterization of when separate item pricing, or grand bundling extracts the maximum revenue $\drev(\calF)$ achievable by any bundle-pricing. On the positive side, we showed that for i.i.d. distributions with support size 2,~the maximum revenue $\rev(\calF)$ achievable by any lottery pricing can always be achieved by a discounted item-pricing, and we can compute %an optimal pricing 
  it in polynomial time. The problem can be also solved in polynomial time for a constant number of items.

A number of interesting problems present themselves. First, we know from Babaioff et al. \cite{Bab14} that discounted item-pricing always achieves a constant fraction (at least 1/6th) of the maximum revenue; what is the constant that can always be guaranteed with respect to the deterministic and randomized maximum revenue? Second, we know that we can compute efficiently an optimal item pricing, and it can be shown that we can also compute an $(1-\epsilon)$-approximately optimal grand bundle price; can we compute efficiently an $(1-\epsilon)$-approximately optimal discounted item pricing? 
(We believe this is the case.) 
Third, besides extending simple item-pricing with the grand bundle,
it is more generally natural to offer discounts on
disjoint groups of items, as in partition mechanisms.
How powerful are such partitioned discounted item-pricings, and
can we compute efficiently an $(1-\epsilon)$-approximately optimal solution of this type?
Finally, regarding i.i.d. distributions, we know that randomization can increase  the revenue for support size 3 in some cases (an example is given by Hart and Nisan \cite{HN12}). Are simple schemes able to extract (approximately) the maximum revenue $\drev(\calF)$ achievable by any bundle pricing for general i.i.d. distributions?

%\begin{acks}
%	
%%	Grants, scholarships, etc
%	
%	
%\end{acks}

% Bibliography
\bibliographystyle{amsalpha}
\bibliography{refs}
%
%\appendix
%\input{appendix}

\end{document}